\documentclass[10pt,twocolumn,journal]{IEEEtran}
\usepackage{amsfonts}
\usepackage{amsbsy}
\usepackage{amsmath,amssymb}
\usepackage{amsthm}
\usepackage{times}
\usepackage{graphicx}
\usepackage{stfloats}
\usepackage{enumerate}
\usepackage[usenames]{color}
\usepackage[dvips]{pstcol}
\usepackage{epstopdf}
\usepackage{cite}
\usepackage{amsmath}
\usepackage{amssymb}
\usepackage{amsfonts}
\usepackage{graphicx}
\usepackage{epsfig}
\usepackage{psfrag}
\usepackage{xcolor}
\usepackage{amsfonts, bm}
\usepackage{epstopdf}
\usepackage{cite}
\usepackage{color}
\usepackage{xcolor}
\usepackage{verbatim}
\usepackage{subfig}
\usepackage{algorithm}
\usepackage{algorithmic}
\newtheorem{theorem}{Theorem}

\input epsf

\begin{document}
\newtheorem{lemma}{Lemma}
\newtheorem{proposition}{Proposition}

\title{Secrecy Rate Maximization of RIS-assisted SWIPT Systems: A Two-Timescale Beamforming Design Approach} 
\author{Ming-Min Zhao,
  \IEEEmembership{Member,~IEEE,} Kaidi Xu, Yunlong Cai,
  \IEEEmembership{Senior Member,~IEEE}, Yong Niu,
  \IEEEmembership{Member, IEEE,} and Lajos Hanzo,
  \IEEEmembership{Fellow, IEEE} \thanks{ M. M. Zhao and Y. Cai
    are with the College of Information Science and Electronic
    Engineering, Zhejiang University, Hangzhou 310027, China, and also with the Zhejiang Provincial Key Laboratory of Information Processing, Communication and Networking (IPCAN), Hangzhou 310027, China (e-mail:
    zmmblack@zju.edu.cn; ylcai@zju.edu.cn). K. Xu is with the Department of Electrical and Electronic Engineering at Imperial College London, SW7 2BX London, U.K. (e-mail: k.xu21@imperial.ac.uk). Y. Niu
    is with the State Key Laboratory of Rail Traffic Control and
    Safety, Beijing Jiaotong University, Beijing 100044, China
    (e-mail: niuy11@163.com). L. Hanzo is with the Department of Electronics and Computer Science,
    University of Southampton, U.K. (e-mail: lh@ecs.soton.ac.uk). \emph{(Corresponding author: Yunlong Cai)}

The work of M. M. Zhao was supported in part by the National Key R\&D Program of China under Grant 2021YFA1003304, in part by the National Natural Science Foundation of China under Grant 62001417, and in part by the Zhejiang Provincial Natural Science Foundation of China under Grant LQ20F010010. 
The work of Y. Cai was supported in part by the National Natural Science Foundation of China under Grants 61971376 and 61831004, and the Zhejiang Provincial Natural Science Foundation for Distinguished Young Scholars under Grant LR19F010002. 
L. Hanzo would like to acknowledge the financial support of the 
Engineering and Physical Sciences Research Council projects EP/W016605/1 
and EP/P003990/1 (COALESCE) as well as of the European Research 
Council's Advanced Fellow Grant QuantCom (Grant No. 789028).
} }

\maketitle

\begin{abstract}
Reconfigurable intelligent surfaces (RISs) achieve high passive beamforming gains for signal enhancement or interference nulling by dynamically adjusting their reflection coefficients. Their employment is particularly appealing for improving both the wireless security and the efficiency of radio frequency (RF)-based wireless power transfer. Motivated by this, we conceive and investigate a RIS-assisted secure simultaneous wireless information and power
transfer (SWIPT) system designed for information and power transfer from a base station (BS) to an information user (IU) and to multiple energy users
(EUs), respectively. Moreover, the EUs are also potential eavesdroppers that may overhear the communication between the BS and IU. We adopt \emph{two-timescale} transmission for reducing the signal processing complexity as well as channel training overhead, and {\color{black}aim for maximizing the average worst-case secrecy rate achieved by the IU}. This is achieved by jointly optimizing the \emph{short-term} transmit beamforming vectors at the BS (including information and energy beams) as well as the \emph{long-term} phase shifts at the RIS, under the energy harvesting constraints considered at the EUs and the power constraint at the BS. The stochastic optimization problem formulated is non-convex with intricately coupled variables, and is non-smooth due to the existence of multiple EUs/eavesdroppers. No standard optimization approach is available for this challenging scenario. To tackle this challenge, we propose a smooth approximation aided stochastic successive convex approximation (SA-SSCA) algorithm. Furthermore, a low-complexity heuristic algorithm is proposed for reducing the computational complexity without unduly eroding the performance. Simulation results show the efficiency of the RIS in securing SWIPT systems. The significant performance gains achieved by our proposed algorithms over the relevant benchmark schemes are also demonstrated.
\end{abstract}


\section{Introduction}
\subsection{Motivation and State-of-the-art}
The increasing growth in the number of wireless Internet-of-Things (IoT) devices, as well as their thirst for ubiquitous communication connectivity and perpetual energy supply continues to inspire researchers to conceive ever more spectral- and energy-efficient solutions. To tackle this challenge, simultaneous wireless information and power transfer (SWIPT) has attracted a huge upsurge of interest, and has been studied comprehensively as a promising solution, where the dual use of radio frequency (RF) signals is exploited \cite{Zhang2013_SWIPT, Lu2015, Zhao2016SWIPT, Zeng2017}. However, an energy user (EU) usually requires much higher received power than that required by an information user (IU), hence the efficiency of wireless
power transfer (WPT) for EUs is considered as the performance bottleneck of practical SWIPT systems. Furthermore, as the EUs are typically deployed close to the base station (BS) for efficient energy reception, this also causes severe wireless security issues since the information users (IUs) can be located farther from the BS and the information intended for the IUs can be easily intercepted by the potential malicious EUs, who tend to have access to stronger signals than the IUs \cite{Liu2014Secrecy_SWIPT, Shi2015secure_SWIPT, Khandaker2019secure_SWIPT}. Thus, developing appropriate security
measures is a crucial issue for SWIPT systems. To tackle this difficulty, sophisticated methods/strategies have been proposed in the literature, such as high-performance beamforming designs \cite{Liu2014Secrecy_SWIPT, Shi2015secure_SWIPT, Wang2020_secure_SWIPT}, constructive interference exploitation and artificial noise injection\cite{Khandaker2019secure_SWIPT}, etc.

Recently, reconfigurable intelligent surface (RIS)-assisted wireless communication has drawn considerable research attention as it is capable of improving the spectral- and/or energy-efficiency \cite{Wu2019Magazine, Renzo2019, Basar2019, Hou2020_RIS}. By dynamically adjusting the reflection coefficients (including amplitudes and phase shifts) of the massive low-cost passive reflecting elements based on the channel state information (CSI), RIS achieves fine-grained passive beamforming gains and thereby beneficially ameliorates the signal propagation for improving the communication performance. In contrast to the existing advanced wireless technologies such as massive multiple-input multiple-output (MIMO) systems, ultra-dense networkinging (UDN) and millimeter wave (mmWave) communications, which improve the communication performance by employing several active nodes/antennas/RF chains. RISs usually only contain passive elements that incur lower hardware/energy cost and less interference contamination. Moreover, RISs are also quite different
from conventional relays. Although they cannot amplify or regenerate the received signal, they achieve full-duplex signal reflection both without self-interference and without processing noise. These compelling advantages of RIS have spurred a lot of interest recently in developing prototypes~\cite{Tan2016ICC, arun2019rfocus, Tang2019ChinaCom, Dai2020access}. Furthermore, numerous global companies, such as NTT DoCoMo, Greenerwave, and Pivotal Commware, are focusing on commercializing RIS-type technologies to reduce the cost, size, weight, and power consumption of existing systems.

The new research paradigm of RIS-aided wireless communication has motivated numerous innovative efforts devoted to their design, and diverse aspects of RISs have been investigated, such as passive beamforming designs \cite{Wu2018_journal, Huang2019, Han2019TVT, zhao2019intelligent, zhao2020TTS}, channel estimation \cite{wang2019channel, You2020JSAC, wei2021channel, Zhou2021_RIS_CE, Ma2021_RIS_CE}, RIS-aided orthogonal frequency division multiplexing (OFDM) systems \cite{Yang2020OFDM_IRS}, RIS-aided mmWave communications \cite{Wang2020mmWave}, physical layer security \cite{Guan2020Secure, yu2019robust, Shen2019secure, WangIRS_Secure_NOMA}, SWIPT \cite{wu2019joint, Pan2020_IRS_SWIPT,Tang2020IRS_SWIPT}, and so on.  In particular, the authors of~\cite{Liu_secure_SWIPT_IRS, Zhou2020IRSSecure_SWIPT} exploited the potential of RIS in improving the SWIPT performance as well as guaranteeing the information security. Specifically, in \cite{Liu_secure_SWIPT_IRS}, artificial noise was injected by the BS and an alternating optimization based iterative algorithm was proposed for maximizing the energy efficiency under specific energy harvesting constraints and minimum (maximum) signal-to-interference-plus-noise ratio (SINR) constraints of the IUs (EUs). In \cite{Zhou2020IRSSecure_SWIPT}, angle-aware user cooperation scheme was adopted, and a two-phase transmission protocol was considered. Specifically, in the first phase the BS avoids direct signal transmission to the attacked user and the other users harvest RF energy by exploiting the popular power-splitting technology. Then in the second phase, the other users cooperate to forward useful signals to the attacked user with the assistance of a RIS. While the above-mentioned authors focus on exploiting RISs for enhancing the
wireless information/energy transmission performance and for ensuring physical layer security, the full knowledge of the instantaneous CSI
(either perfect or imperfect) including the BS-RIS and RIS-user links is required. In practice, however, since the number of reflecting elements is usually large, obtaining the full instantaneous CSI is a challenge, which may incur higher channel training overhead, and using excessive time for channel estimation will lead to reduced user transmission rate due to the limited time left for data transmission. Besides, designing the RIS reflection coefficients in each channel realization based on the full instantaneous CSI also increases the signal processing complexity, since usually iterative algorithms are required for RIS optimization \cite{Liu_secure_SWIPT_IRS, Zhou2020IRSSecure_SWIPT}.

\subsection{Main Contributions}
Against the above background, we study a RIS-assisted secure SWIPT system and adopt the \emph{two-timescale} transmission protocol \cite{zhao2019intelligent,zhao2020TTS} for reducing the signal processing complexity and channel training overhead. Specifically, a RIS is deployed to assist in the information and power transfer from a multi-antenna BS to a single-antenna IU and multiple single-antenna EUs, respectively, while in the meantime guaranteeing the information security of the IU.\footnote{{\color{black}Note that the proposed algorithm is still applicable also for the more general  case associated with multiple IUs}, but here we consider the scenario of single-antenna users. Further extension to the multi-antenna case is set aside for our future work.}

\begin{itemize}
\item To this end, {\color{black}we maximize the average worst-case secrecy rate achieved by the IU upon jointly optimizing the \emph{short-term} transmit beamforming vectors at the BS (including the information and energy beams) and \emph{long-term} phase shifts at the RIS, under the energy harvesting constraints at the EUs and the total transmit power constraint at the BS.}

\item We design the \emph{long-term} RIS reflection phase shifts based on the channel statistics, while the \emph{short-term} transmit beamforming vectors at the BS are designed based on the effective instantaneous CSI with fixed RIS phase shifts (the corresponding channel dimension is much smaller than that of the full instantaneous CSI). Note that the solution structure of the short-term beamforming design problems is taken into consideration during long-term phase-shift optimization, thus the long-term and short-term variables are jointly optimized.

\item The stochastic optimization problem formulated is non-convex and non-smooth with intricately coupled variables, which is difficult to solve, and the existing algorithms are not applicable. To tackle this difficulty, we propose a smooth approximation aided stochastic successive convex approximation (SA-SSCA) algorithm that can be deployed in an online fashion, where the original problem is first approximated by a smooth problem and then decomposed into a long-term subproblem and a family of short-term subproblems.

  \item We propose an efficient concave-convex procedure based block coordinate descent (CCCP-BCD) algorithm that
  is integrated in the SA-SSCA algorithm to solve the short-term subproblems and prove its convergence to stationary solutions. Then, we show that the long-term subproblem can be easily solved and its closed-form optimal solution is available. Besides, a heuristic algorithm is presented, which can achieve performance comparable to the SA-SSCA algorithm, but with lower complexity.

  \end{itemize}
  
\subsection{Organization}
The rest of the paper is organized as follows. In Section \ref{sec_system_model}, we present the system model and problem formulation. In Section \ref{sec_proposed_algorithm}, we propose an efficient SA-SSCA algorithm to solve the formulated problem. In Section \ref{section_low_complexity}, a low-complexity heuristic algorithm is further proposed. In Section \ref{Section_Simulation}, numerical results are provided to evaluate the performance of the proposed algorithms. Finally, we conclude the paper in Section \ref{sec_conclusions}.

\emph{Notations:} Scalars, vectors and matrices are respectively denoted by lower/upper case, boldface lower case and boldface upper case letters. For an arbitrary matrix $\mathbf{A}$, $\mathbf{A}^T$, $\mathbf{A}^*$, $\mathbf{A}^{H}$ and $\mathbf{A}^{\dagger} $ denote its transpose, conjugate, conjugate transpose and pseudo-inverse, respectively. $\|\cdot\|$ denote the Euclidean norm of a complex vector, and $|\cdot|$ denotes the absolute value of a complex scalar. $\mathbf{a} \circ \mathbf{b}$ denotes the element-wise product of two vectors. $\mathcal{CN}(\mathbf{x},\bm{\Sigma})$ denotes the distribution of a circularly symmetric complex Gaussian (CSCG) random vector with mean vector $\mathbf{x}$ and covariance matrix $\bm{\Sigma}$; and $\sim$ stands for ``distributed as''. For given numbers $x_1,\cdots,x_N$, $\textrm{diag}(x_1,\cdots,x_N)$ denotes a diagonal matrix with $\{x_1,\cdots,x_N\}$ being its diagonal elements and $\textrm{diag}(\mathbf{A})$ denotes a vector which contains the diagonal elements of matrix $\mathbf{A}$. The letter $\jmath$ is used to represent $\sqrt{-1}$. {\color{black}$\mathcal{F}^N$ is defined as the Cartesian product of $N$ identical sets each given by $\mathcal{F}$.} For a complex number $x$, $\Re \{x\}$ denotes its real part and $\angle x$ denotes its angle. $\mathbf{I}$ and $\mathbf{0}$ denote an identity matrix and an all-zero vector with appropriate dimensions, respectively. $\mathbb{E}\{\cdot\}$ denotes the statistical expectation. For any real number $x$, we define $[x]^+\triangleq \max (x,0)$. $\mathbb{C}^{n\times m}$ denotes the space of $n\times m$ complex matrices.

\section{System Model and Problem Formulation} \label{sec_system_model}
\subsection{System Model}
As shown in Fig. \ref{fig:system_model}, we consider a RIS-aided wireless system which consists of one BS equipped with $N_s$ antennas, one RIS composed of $N_r$ reflecting elements, one IU, and $M$ EUs denoted by $\mathcal{M} \triangleq \{ 1, \cdots,M\}$. Note that the EUs can potentially overhear  the confidential information sent from the BS to the IU. All the users are equipped with a single antenna and the EUs are assumed to be deployed in more proximity to  the BS for ease of energy harvesting due to severe channel attenuation. The RIS is attached to a smart controller that is able to communicate with the BS via a separate backhaul link for coordinating transmission and exchanging information, such as CSI and RIS phase shifts \cite{Wu2019Magazine}. In this work, we assume that the RIS is deployed in the vicinity of the EUs such that it cannot only help increase the energy harvested by the EUs, but also assist in enhancing the signal strength for the distant IU and improving its security through careful phase-shift optimization. Besides, since the signal transmitted through the BS-RIS-user link suffers from the double path loss, the signals reflected by RIS two or more times are ignored \cite{Wu2019Magazine}.

\begin{figure}[!t]
	\centering
		\setlength{\abovecaptionskip}{-0.2cm}
	\setlength{\belowcaptionskip}{-0.2cm}
	\scalebox{0.5}{\includegraphics{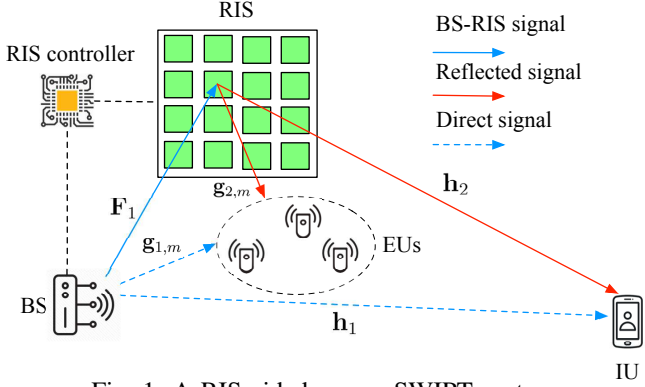}}
	\caption{A RIS-aided secure SWIPT system.}\label{fig:system_model}
\end{figure}

For simplicity, we consider linear transmit precoding at the BS and assume that each IU/EU is assigned with one individual information/energy beam without loss of generality.\footnote{Note that the energy beams carry no information but instead pseudo-random signals that are assumed to be unknown at both the BS and the IU, therefore they can also be viewed as artificial noise signals that can help improve the secrecy performance.} Hence, the transmitted signal at the BS can be expressed as
\begin{equation}
	\mathbf{x}=\mathbf{w}s^{\textrm{I}}+\sum_{m\in \mathcal{M}}\mathbf{p}_ms_m^{\textrm{E}},
\end{equation}
where $s^{\textrm{I}}$ and $s^{\textrm{E}}_m$ denote the information-bearing and energy-carrying signals, respectively, while $\mathbf{w}\in\mathbb{C}^{N_s\times 1}$ and $\mathbf{p}_m\in \mathbb{C}^{N_s\times 1}$ are the information and energy beamforming vectors for the IU and EU $m$, respectively. In particular, since the energy-carrying signals $s^{\textrm{E}}_m$'s do not carry any information, they are assumed to be independently generated from an arbitrary distribution with $\mathbb{E}\{|s^{\textrm{E}}_m|^2\}=1, \forall m \in \mathcal{M}$. In contrast, the information signal $s^{\textrm{I}}$ is assumed to be a CSCG random variables with zero mean and unit variance, i.e., $s^{\textrm{I}} \sim \mathcal{CN}(0,1)$.  Thus, the transmit power required at the BS is given by
\begin{equation}
	\mathbb{E}\{\mathbf{x}^H\mathbf{x}\}=\|\mathbf{w}\|^2+\|\mathbf{P}\|^2,
\end{equation}
where $\mathbf{P}\in \mathbb{C}^{N_s\times M}\triangleq [\mathbf{p}_1,\mathbf{p}_2,...,\mathbf{p}_M]$.

Let us define $\mathbf{h}_1^H\in\mathbb{C}^{1\times Ns}$, $\mathbf{F}_{1}^H\in \mathbb{C}^{Nr\times Ns}$ and $\mathbf{g}_{1,m}^H\in\mathbb{C}^{1\times Ns}$ as the channels from the BS to the IU, RIS and EU $m$, respectively, and let $\mathbf{h}_2^H\in \mathbb{C}^{1\times Nr}$ and $\mathbf{g}_{2,m}^H\in\mathbb{C}^{1\times N_r}$ denote the channels from the RIS to the IU and EU $m$, respectively. Denote by $\boldsymbol{\Theta}\triangleq \text{diag}(\phi_1 = \beta_1e^{j\theta_1},...,\phi_N = \beta_{N_r}e^{j\theta_{N_r}})=\text{diag}(\boldsymbol{\phi})$ ($\bm{\phi} \triangleq [\phi_1,\cdots,\phi_N]$) as the reflection-coefficient matrix at the RIS, where $\beta_n\in[0,1]$ and $\theta_n\in[0,2\pi)$ ($n \in \mathcal{N} \triangleq \{ 1,\cdots, N_r\}$) are the reflection amplitude and phase shift of the $n$-th element, respectively. In order to reduce the implementation cost of RIS, we set $\beta_n=1,\forall n$ to maximize the signal reflection \cite{Wu2018_journal}. Hence, the signal received at the IU can be expressed as
\begin{equation}
	y^{\textrm{I}}=\tilde{\mathbf{h}}^H\mathbf{x}+n_{I},
\end{equation}
where $n_{I} \sim \mathcal{CN}(0,\sigma_{I}^2)$ is the additive white Gaussian noise (AWGN) at the IU and $\tilde{\mathbf{h}}^H\triangleq (\mathbf{h}_2^H\boldsymbol{\Theta}^H\mathbf{F}_1^H + \mathbf{h}_1^H)=(\boldsymbol{\phi}^H\text{diag}(\mathbf{h}_2^H)\mathbf{F}_1^H + \mathbf{h}_1^H)$ is the effective channel from the BS to the IU. Similarly, the signal received at EU $m$ is given by
\begin{equation}
	y_m^{\textrm{E}}=\tilde{\mathbf{g}}^H_m\mathbf{x}+n_{E,m},
\end{equation}
where $n_{E,m} \sim \mathcal{CN}(0,\sigma_{E,m}^2)$ is the AWGN at EU $m$ and $\tilde{\mathbf{g}}^H_m\triangleq(\mathbf{g}_{2,m}^H\boldsymbol{\Theta}^H\mathbf{F}_1^H + \mathbf{g}_{1,m}^H)=(\boldsymbol{\phi}^H\text{diag}(\mathbf{g}_{2,m}^H)\mathbf{F}_1^H + \mathbf{g}_{1,m}^H)$ denotes the effective channel from the BS to EU $m$.
Accordingly, the SINRs at the IU and EUs are given by
\begin{equation} \label{SINR_IU}
	\textrm{SINR}^{\textrm{I}} (\mathbf{w}, \mathbf{P}, \bm{\Theta})=\frac{|\tilde{\mathbf{h}}^H\mathbf{w}|^2}{\|\tilde{\mathbf{h}}^H\mathbf{P}\|^2+\sigma_{I}^2},
\end{equation}
\begin{equation}
	\textrm{SINR}^{\textrm{E}}_m (\mathbf{w}, \mathbf{P}, \bm{\Theta})=\frac{|\tilde{\mathbf{g}}_m^H\mathbf{w}|^2}{\|\tilde{\mathbf{g}}_m^H\mathbf{P}\|^2+\sigma_{E,m}^2},\;\forall m \in \mathcal{M},
\end{equation}
respectively.\footnote{In the case that the energy beams are assumed to be known at both the BS and the IU before data transmission \cite{Xu2014}, then the interference caused by these energy beams can be cancelled at the IU and the SINR expression in \eqref{SINR_IU} can be simplified as $\frac{|\tilde{\mathbf{h}}^H\mathbf{w}|^2}{\sigma_{I}^2}$. Since the proposed algorithm can be directly applied to solve the simplified problem in this case, the details are omitted for brevity.}

On the other hand, by ignoring the noise power, the received RF power at EU $m$ can be written as
\begin{equation}
	Q_m (\mathbf{w}, \mathbf{P}, \bm{\Theta})=|\tilde{\mathbf{g}}_{m}^H\mathbf{w}|^2+\|\tilde{\mathbf{g}}_{m}^H\mathbf{P}\|^2, \;\forall m \in \mathcal{M}.
\end{equation}

\subsection{Problem Formulation}
In this paper, {\color{black}we aim for maximizing the average worst-case secrecy rate achieved by the IU by jointly optimizing the transmit beamforming vectors at the BS (including the information and energy beams) and the phase shifts at the RIS, subject to the total transmit power constraint at the BS and energy harvesting constraints at the EUs.} We adopt a similar TTS transmission protocol in \cite{zhao2019intelligent, zhao2020TTS} to reduce the signal processing complexity and training/signaling overhead of the conventional instantaneous CSI based schemes. {\color{black}Note that similar transmission protocols are also adopted by \cite{Zhi2021ECL, zhi2021twotimescale, zhi2021ergodic} for RIS-aided massive MIMO systems.} Specifically, as shown in Fig. \ref{fig:timeline}, the time axis is divided into several super-frames within which the statistics (distributions) of all channels are assumed to be constant. Each super-frame is composed of several ($T_f$) frames and each frame further consists of multiple ($T_s$) time slots, and all channels are assumed to remain approximately constant within each time slot and vary over different time slots, i.e., the quasi-static flat-fading model is assumed for all channels. We assume that the effective instantaneous CSI at each time slot, i.e., $\tilde{\mathbf{h}}$ and $\{\tilde{\mathbf{g}}_m\}$ with given $\bm{\phi}$ (or equivalently $\bm{\Theta}$) can be obtained at the beginning of each time slot and $T_c$ (typically less than the number of time slots contained in each frame) channel samples (possibly outdated) of the BS-RIS, BS-user and RIS-user links, i.e., $\mathcal{H}_j \triangleq \{\mathbf{h}_{1,j},\mathbf{g}_{1,m,j}, \text{diag}(\mathbf{h}_{2,j}^H)\mathbf{F}_{1,j}^H, \text{diag}(\mathbf{g}_{2,m,j}^H)\mathbf{F}_{1,j}^H \}$ ($j$ denotes the sample index), $j=1,\cdots,T_c$, are available at the BS at the end of each frame. {\color{black}The $T_c$ channel samples can be generated according to the channel's statistical information if we assume that there is a statistical CSI estimation phase before data transmission \cite{zhao2019intelligent}.}\footnote{\color{black}Since the statistical CSI has been estimated in this phase, an arbitrary number of channel samples can be generated for long-term RIS phase-shift optimization without imposing unaffordable channel estimation overhead. Note that these channel samples are random realizations of the statistical CSI, and by generating them, no new information is gleaned about the channel.} Therefore, we aim to optimize the long-term passive RIS phase-shift matrix $\bm{\Theta}$ based on the channel samples of all links (or implicitly, the channel statistics), while the short-term active beamforming vectors $\mathbf{w}$ and $\mathbf{P}$ at the BS are designed to adapt to the instantaneous CSI of the users' effective fading channels $\tilde{\mathbf{h}}$ and $\{\tilde{\mathbf{g}}_m\}$ with fixed RIS phase shifts over the time slots within each frame. Besides, since only some channel samples are needed in each frame and the estimation of the effective channels is much easier than obtaining the real-time full  instantaneous CSI in each time slot, the channel estimation overhead can be significantly reduced.

\begin{figure*}
\centering
\scalebox{1.25}{\includegraphics{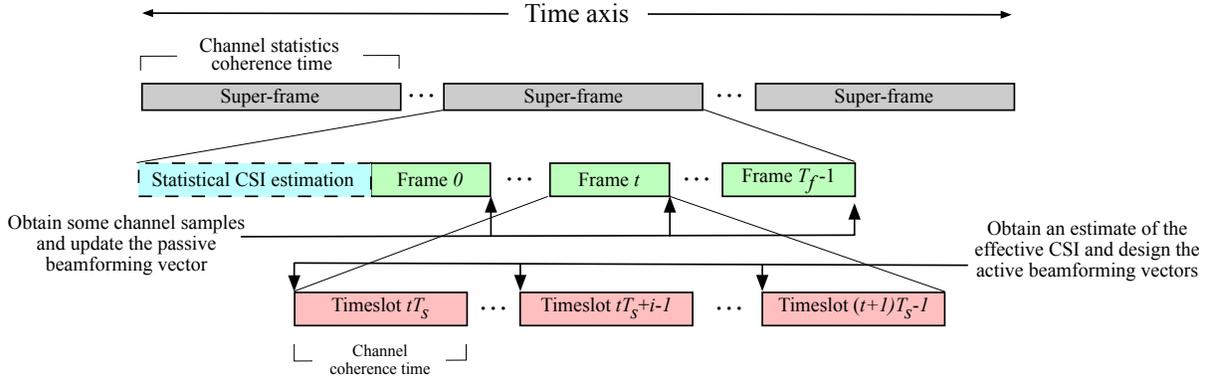}}
\caption{\color{black}Frame structure of the two-timescale scheme.}\label{fig:timeline}
\end{figure*}

%

Let $\boldsymbol{\theta}=[\theta_1, ..., \theta_{N_r}]^T=\angle\text{diag}(\boldsymbol{\Theta})$, then the worst-case secrecy rate achieved by the IU is given by\footnote{The notations $f(\cdot, \bm{\Theta})$ and $f(\cdot, \bm{\theta})$, where $f(\cdot)$ denotes an arbitrary function, will be used interchangeably in the following when there is no ambiguity.}
\begin{equation}
	\begin{aligned}
	S(\mathbf{w}, \mathbf{P}, \bm{\theta},\mathcal{H}) = & \Big[\log(1+\textrm{SINR}^{\textrm{I}}(\mathbf{w}, \mathbf{P}, \bm{\theta})) \\
&	-\max_{m\in\mathcal{M}} \log(1+\textrm{SINR}_m^{\textrm{E}}(\mathbf{w}, \mathbf{P}, \bm{\theta}))\Big]^+,
	\end{aligned}
\end{equation}
where $\mathcal{H} \triangleq \{\mathbf{h}_{1},\mathbf{g}_{1,m}, \text{diag}(\mathbf{h}_{2}^H)\mathbf{F}_{1}^H, \text{diag}(\mathbf{g}_{2,m}^H)\mathbf{F}_{1}^H \}$.
Accordingly, the optimization problem is mathematically formulated as follows:
\begin{subequations} \label{orignal_prob}
\begin{align}
	\max_{\boldsymbol{\theta}} \; &{\color{black} \mathbb{E} \bigg\{\max_{\mathbf{w},\;\mathbf{P} }  \; S(\mathbf{w}, \mathbf{P}, \bm{\theta},\mathcal{H}) \bigg\}} \label{objectve_function} \\
	\textrm{s.t.}\; & \|\mathbf{w}\|^2+\|\mathbf{P}\|^2\leq P_t,  \label{power_constraint}\\
	& Q_m (\mathbf{w}, \mathbf{P}, \bm{\theta}) \geq \epsilon_m,\;\forall m \in \mathcal{M}, \label{EH_constraint} \\
	&0\leq\theta_n < 2\pi, \forall n \in \mathcal{N},\label{phase_constraint}
\end{align}
\end{subequations}
where {\color{black}the expectation is taken over all channels' random realizations within the super-frame considered}, $P_t$ denotes the BS power budget and $ \epsilon_m$ represents the minimum RF receive power requirement of EU $m$. By solving problem \eqref{orignal_prob}, we aim to seek a reflection phase-shift vector $\bm{\theta}$ that is statistically optimal in the long term (by solving the outer problem in \eqref{orignal_prob}) and a short-term beamforming design that only requires the effective/partial instantaneous CSI $\tilde{\mathbf{h}}$ and $\{\tilde{\mathbf{g}}_m\}$ with given $\bm{\theta}$ (by solving the inner problem in \eqref{orignal_prob}). 

Let $\Omega \triangleq \{  \mathbf{w}(\mathcal{H}), \mathbf{P}(\mathcal{H}), \forall \mathcal{H} \}$ denote the collection of the short-term optimization variables for all possible channel realizations that satisfy the transmit power constraint \eqref{power_constraint}, and let $\mathbb{E}\{ S(\mathbf{w}(\mathcal{H}), \mathbf{P}(\mathcal{H}), \bm{\theta}, \mathcal{H}) \}$ denote the average worst-case secrecy rate. Then, we observe that problem \eqref{orignal_prob} is difficult to solve because 1) it is a non-convex stochastic optimization problem with the transmit beamforming vectors and RIS phase shifts intricately coupled in the objective function; 2) a closed-form
expression of $\mathbb{E}\{ S(\mathbf{w}(\mathcal{H}), \mathbf{P}(\mathcal{H}), \bm{\theta}, \mathcal{H}) \}$ with given either $\bm{\theta}$ or $\Omega$, is difficult to obtain in general; 3) as multiple EUs (potential eavesdroppers) are considered and these exist a $[\cdot]^+$ operation in \eqref{objectve_function}, the worst-case secrecy rate $S(\mathbf{w}, \mathbf{P}, \bm{\theta},\mathcal{H})$ is non-smooth.\footnote{\color{black}Although this contribution and our previous treatises~\cite{zhao2019intelligent, zhao2020TTS} all consider the two-timescale transmission scheme, the main difference lies in the non-smooth nature of the objective function. How to tackle this challenge is one of the main contributions of this work, which will become more explicit, as we progress.}  Generally, there is no standard method for solving such non-convex non-smooth stochastic optimization problems optimally. In the next two sections, we propose efficient algorithms to solve problem \eqref{orignal_prob} sub-optimally.
\newtheorem{remark}{Remark}
\begin{remark}
{\color{black}	\emph{As long as the statistical CSI is available, the long-term optimization can either be done at the end of the statistical CSI estimation phase in an offline fashion, or at the end of each frame in an online fashion. In the offline case, the computational burden and memory cost at the end of the statistical CSI estimation phase might be high since a number of iterations are required and the number of channel samples that should be stored is large. {\color{black}In this case, the long-term reflection coefficients are fixed in the rest of the frames, and in practice they can be initialized with those in the last super-frame for ensuring that the proposed   algorithm can converge promptly, since the channel statistics in two adjacent super-frames are quite similar.} While for the online case, only a small number of channel samples are utilized to update the long-term variables, and thus only a small number of short-term subproblems are required to be solved at each iteration, which also means that the reflection coefficients are gradually refined in each frame. {\color{black}This is similar to the updating process of neural network parameters in deep reinforcement learning. By doing so, the algorithm can automatically track the slowly-varying channel statistics.} Therefore, online algorithms are more desirable due to its lower memory requirement and complexity.}}
\end{remark} 
\begin{remark}
\emph{The efficiency of the proposed scheme depends on how fast the channel statistics vary. If the statistical CSI changes slowly, e.g., when the users have low mobility, then the implementation cost can be significantly reduced, since only the effective channels have to be estimated in each time slot and RIS optimization only occures sporadically. By contrast, if the channel statistics rapidly become outdated, then the overhead required by statistical CSI estimation will increase, and so will the frequency of RIS optimization, which affects the the proposed scheme's efficiency. {\color{black}In this paper, we only focus our attention on a quasi-static statistical CSI scenario, while the effects of time-variant statistical CSI are studied via numerical simulations (in Section \ref{section_time_variant_S_CSI}). Further investigation of rapidly time-variant statistical CSI scenarios is set aside for our future research.}}
\end{remark}

\section{Proposed SA-SSCA Algorithm} \label{sec_proposed_algorithm}
In this section, we propose the SA-SSCA algorithm to efficiently solve problem \eqref{orignal_prob}. Specifically, we first transform problem \eqref{orignal_prob} into a more tractable form by employing smooth approximation, which turns \eqref{objectve_function} into a differentiable function. Then, the SSCA algorithm is proposed to solve the transformed problem by decomposing it into a long-term subproblem and a family of short-term subproblems. We show that these two types of subproblems can be efficiently solved, and the proposed SA-SSCA algorithm is guaranteed to converge. {\color{black}Note that although the SSCA framework of~\cite{Liu2018TOSCA} is adopted, it is nontrivial to design efficient long-term and short-term optimization algorithms, especially when the objective function and constraints are in complex forms as in our case.}

\subsection{Problem Transformation and Decomposition}
First, we propose to ignore the operator $[\cdot]^+$ in the objective function \eqref{objectve_function} since if the worst-case secrecy rate $S(\mathbf{w}, \mathbf{P}, \bm{\theta},\mathcal{H})$ is negative at an arbitrary time slot, we can always set the corresponding information-bearing signal $s^{\textrm{I}}$ equal to $ s^{\textrm{E}}_m$ ($m$ can be chosen arbitrarily) such that no information is sent and $S(\mathbf{w}, \mathbf{P}, \bm{\theta},\mathcal{H})=0$ is satisfied without violating the constraints \eqref{power_constraint} and \eqref{EH_constraint}. As a result, problem \eqref{orignal_prob} can be equivalently transformed into
\begin{equation}
\begin{aligned}
	\max_{\boldsymbol{\theta}} \; &\mathbb{E}\bigg\{\max_{\mathbf{w},\; \mathbf{P} } \;  \tilde{S} (\mathbf{w}, \mathbf{P}, \bm{\theta},\mathcal{H}) \bigg\}\\
	\textrm{s.t.}\; & \eqref{power_constraint}-\eqref{phase_constraint},
\end{aligned}\label{trans_prob}
\end{equation}
where $\tilde{S}(\mathbf{w}, \mathbf{P}, \bm{\theta},\mathcal{H}) \triangleq \log(1+\textrm{SINR}^{\textrm{I}}(\mathbf{w}, \mathbf{P}, \bm{\theta}))-\max_{m\in\mathcal{M}} \log(1+\textrm{SINR}_m^{\textrm{E}}(\mathbf{w}, \mathbf{P}, \bm{\theta}))$. Note that problem \eqref{orignal_prob} and problem \eqref{trans_prob} have the same optimal solution since $\tilde{S}(\mathbf{w}, \mathbf{P}, \bm{\theta},\mathcal{H})$ is always non-negative as discussed above.

Next, we can see that problem \eqref{trans_prob} is still difficult to solve due to its non-smooth objective function caused by the existence of multiple EUs.
To tackle this difficulty, we introduce the following smooth approximation lemma to construct a differentiable and continuous approximation of $\tilde{S}(\mathbf{w}, \mathbf{P}, \bm{\theta},\mathcal{H}) $.

\begin{lemma}\cite{Boyd2004}\label{lemma1}
Let $f(\mathbf{x})\triangleq \max(\mathbf{x})$ and $\hat{f}_p(\mathbf{x})\triangleq \frac{1}{p}\text{ln}\left( \sum_{i=1}^m e^{px_i}\right) $, where $\mathbf{x}=  [x_1, x_2, ... ,x_m]^T\in\mathbb{R}^{m \times 1}$, $m$ is a positive integer and $p$ is a real number. Then, we have
\begin{equation}
	f(\mathbf{x}) \leq \hat{f}_p(\mathbf{x}),
\end{equation}
and
\begin{equation}
	\lim_{p\to \infty} \hat{f}_p(\mathbf{x})=f(\mathbf{x}),
\end{equation}
when $m\geq 2$ and $p>0$. {\color{black}The approximation gap is upper-bounded by $\frac{1}{p}\log m$.}
\end{lemma}
\noindent
Lemma \ref{lemma1} indicates that the max function can be approximated by a log-sum-exp function, which is convex and differentiable (in fact, analytical). By resorting to Lemma \ref{lemma1} and the fact that $1+\textrm{SINR}_{m}^{\textrm{E}}(\mathbf{w}, \mathbf{P}, \bm{\theta})>0, \forall m \in \mathcal{M}$, $\tilde{S}(\mathbf{w}, \mathbf{P}, \bm{\theta},\mathcal{H})$ can be lower-bounded by
\begin{equation}
	\begin{aligned}
\tilde{S}(\mathbf{w}, \mathbf{P}, & \bm{\theta},\mathcal{H}) >  \bar{S}(\mathbf{w}, \mathbf{P}, \bm{\theta},\mathcal{H}) \triangleq  \log(1+\textrm{SINR}^{\textrm{I}} (\mathbf{w}, \mathbf{P}, \bm{\theta})) \\
&-\frac{1}{p}\log\left(\sum_{m\in \mathcal{M}}(1+\textrm{SINR}_m^{\textrm{E}}(\mathbf{w}, \mathbf{P}, \bm{\theta}))^{p}\right).
\end{aligned}
\end{equation}
Hence, problem \eqref{trans_prob} can be further transformed into
\begin{equation}
\begin{aligned}
	\max_{\boldsymbol{\theta}} \; & \mathbb{E}\bigg\{ \max_{\mathbf{w},\; \{\mathbf{p}_m\}}\; \bar{S}(\mathbf{w}, \mathbf{P}, \bm{\theta},\mathcal{H}) \bigg\}\\
	\textrm{s.t.}\;& \eqref{power_constraint}-\eqref{phase_constraint}.
\end{aligned}\label{trans_prob1}
\end{equation}
Note that for any finite value of $p$, problem \eqref{trans_prob1} is an approximation of problem \eqref{trans_prob}, while they become equivalent as $p \rightarrow \infty$.

Then, in order to resolve the difficulty caused by the fact that the long-term variable $\bm{\theta}$ and short-term variables $\{\mathbf{w}, \mathbf{P} \}$ are coupled, we further employ the primal decomposition method \cite{Boyd2004, Liu2018TOSCA} to decompose problem \eqref{trans_prob1} into a long-term subproblem and a family of short-term subproblems. Specifically, for fixed long-term phase shifts $\bm{\theta}$, we have the following short-term subproblem:
\begin{equation} \label{shortterm_problem}
	\begin{aligned}
\mathcal{P}_S(\bm{\theta}, \mathcal{H}): \max_{\mathbf{w},\; \mathbf{P}} \; & \bar{S}(\mathbf{w}, \mathbf{P}, \bm{\theta},\mathcal{H}) \\
		\textrm{s.t.}\; & \eqref{power_constraint}\; \textrm{and}\; \eqref{EH_constraint},
	\end{aligned}
\end{equation}
which corresponds to one channel realization $\mathcal{H}$ and is solved in each time slot. Note that since $\bm{\theta}$ is fixed in $\mathcal{P}_S(\bm{\theta}, \mathcal{H})$, the BS only needs to estimate the low-dimensional effective fading channels $\tilde{\mathbf{h}}$ and $\tilde{\mathbf{g}}_m,\forall m \in \mathcal{M}$, in each time slot, instead of the full instantaneous CSI. Let $(\mathbf{w}^J(\bm{\theta},\mathcal{H}), \mathbf{P}^J(\bm{\theta},\mathcal{H}) )$ denote a stationary solution of $\mathcal{P}_S(\bm{\theta}, \mathcal{H})$, which is obtained by running the short-term sub-algorithm (will be introduced in Section \ref{section_shortterm_algorithm}) for a sufficiently large number ($J$) of iterations.\footnote{Note that typically, the convergence to a stationary solution requires $J \rightarrow \infty$, however, the short-term sub-algorithm always runs for a finite number of iterations in practice. Therefore, $(\mathbf{w}^J(\bm{\theta},\mathcal{H}), \mathbf{P}^J(\bm{\theta},\mathcal{H}) )$ may not be the exact stationary solution of $\mathcal{P}_S(\bm{\theta}, \mathcal{H})$, but this will not affect the convergence of the proposed algorithm as the error caused by this issue will diminish to zero as $J\rightarrow \infty$ \cite{Liu2018TOSCA}.} Then, with the short-term policy $\{\mathbf{w}^J(\bm{\theta},\mathcal{H}), \mathbf{P}^J(\bm{\theta},\mathcal{H}) \}$, we formulate the following long-term subproblem:
\begin{equation}\label{longterm_problem}
	\begin{aligned}
\mathcal{P}_L^J: \max_{\boldsymbol{\theta}}\; & \mathbb{E}\{\bar{S} (\mathbf{w}^J(\bm{\theta},\mathcal{H}), \mathbf{P}^J(\bm{\theta},\mathcal{H}), \bm{\theta}) \}, \\
	\textrm{s.t.}\; &0\leq\theta_n < 2\pi, \forall n \in \mathcal{N}.
	\end{aligned}
\end{equation}
According to \cite{Liu2018TOSCA}, a stationary solution (up to  certain errors caused by the finite $J$) can be found by solving the long-term subproblem $\mathcal{P}_L^J$ to obtain a stationary point $\bm{\theta}^\star$, and finding a stationary solution $(\mathbf{w}^J(\bm{\theta},\mathcal{H}), \mathbf{P}^J(\bm{\theta},\mathcal{H}) )$ of $\mathcal{P}_S(\bm{\theta}, \mathcal{H})$ for each $\mathcal{H}$ (i.e., finding a short-term beamforming design policy). In the following subsections, we will present how to solve $\mathcal{P}_S(\bm{\theta}, \mathcal{H})$ and $\mathcal{P}_L^J$ efficiently.

\subsection{Solving the Short-term Subproblem} \label{section_shortterm_algorithm}
With given RIS phase-shift vector $\bm{\theta}$, we observe that $\mathcal{P}_S(\bm{\theta}, \mathcal{H})$ is still difficult to solve as its objective function $\bar{S}(\mathbf{w}, \mathbf{P}, \bm{\theta},\mathcal{H})$ and the energy harvesting constraint \eqref{EH_constraint} are non-convex. To transform $\mathcal{P}_S(\bm{\theta}, \mathcal{H})$ into a more tractable form, we resort to the following theorem.
\begin{theorem}\label{trans_theorem1}
	Let $z\in\mathbb{C}$, $u \in\mathbb{C}$ and $v \in\mathbb{C}$ denote the introduced auxiliary variables, then the following problem:
	\begin{equation} \label{trans_prob2}
		\begin{aligned}
			\min_{\mathbf{w},\mathbf{P}, z, u, v} \; & ze(\mathbf{w},\mathbf{P},u) \!-\! \log(z) \!+\! \frac{1}{p}v\tilde{e}(\mathbf{w},\mathbf{P}) \!-\! \frac{1}{p}\log(v),\\
			\textrm{s.t.}\; & \eqref{power_constraint}\; \textrm{and}\; \eqref{EH_constraint},
		\end{aligned}
	\end{equation}
where $e(\mathbf{w},\mathbf{P},u)\triangleq |u|^2| \tilde{\mathbf{h}}^H\mathbf{w}|^2+|u|^2\sigma^2+|u|^2\|\tilde{\mathbf{h}}^H\mathbf{P}\|^2+1-2\Re\{u^* \tilde{\mathbf{h}}^H\mathbf{w}\}$ and $\tilde{e}(\mathbf{w},\mathbf{P})\triangleq \sum_{m \in \mathcal{M}}(1+\textrm{SINR}_m^{\textrm{E}}(\mathbf{w}, \mathbf{P}, \bm{\theta}))^{p}$, is equivalent to $\mathcal{P}_S(\bm{\theta}, \mathcal{H})$, in the sense that the global optimal solution $(\mathbf{w}^\star, \mathbf{P}^\star)$ for the two problems are identical.
\end{theorem}
\begin{proof}
Please refer to Appendix \ref{proof_theorem1}.
\end{proof}
\noindent
Theorem \ref{trans_theorem1} implies that maximizing $\bar{S}(\mathbf{w}, \mathbf{P}, \bm{\theta},\mathcal{H})$ can be accomplished via minimizing the objective function of problem \eqref{trans_prob2}, which is a weighted sum of $e(\mathbf{w},\mathbf{P},u)$ and $\tilde{e}(\mathbf{w},\mathbf{P})$.\footnote{Note that the intrinsic idea behind Theorem \ref{trans_theorem1} is similar to the weighted sum mean-square error minimization (WMMSE) transformation proposed in \cite{Shi2011WMMSE}, however, the difference is that these exists a polynomial term of $\textrm{SINR}_m^{\textrm{E}}(\mathbf{w}, \mathbf{P}, \bm{\theta})$ in $\bar{S}(\mathbf{w}, \mathbf{P}, \bm{\theta},\mathcal{H})$ and thus this term cannot be transformed into weighted mean squared error (MSE).} It is noteworthy that problem \eqref{trans_prob2} is in the space of $(\mathbf{w}, \mathbf{P},z,u,v)$ and is easier to handle since optimizing each variable while fixing the others is relatively easy. To proceed, we further introduce a set of auxiliary variables $y_m, m\in \mathcal{M}$, then problem \eqref{trans_prob2} can be equivalently transformed into the following problem:
\begin{subequations} \label{trans_prob3}
	\begin{align}
		\min_{\mathbf{w},\mathbf{P}, z, u, v,\{y_m\}} \; & ze(\mathbf{w},\mathbf{P},u)-\log(z)\\
		&+\frac{1}{p}v\left(\sum_{m\in \mathcal{M}}(1+y_m)^{p}\right)-\frac{1}{p}\log(v) \notag \\
		\text{s.t.}\; & \textrm{SINR}_m^{\textrm{E}} (\mathbf{w}, \mathbf{P}, \bm{\theta}) \leq y_m, \;\forall m \in \mathcal{M}, \label{SINR_const}\\
		&\eqref{power_constraint}\;\textrm{and}\; \eqref{EH_constraint}. \notag
	\end{align}
\end{subequations}
The equivalence between problem \eqref{trans_prob2} and problem \eqref{trans_prob3} can be easily verified as the optimal $y_m^{\star}$ of problem \eqref{trans_prob3} must satisfy $y_m^{\star} = \textrm{SINR}_m^{\textrm{E}} (\mathbf{w}, \mathbf{P}, \bm{\theta})$ (otherwise, we can always slightly decrease $y_m^{\star}$ such that the objection value is reduced without violating any constraints) and by substituting $y_m^{\star}$ back into the objective function of problem \eqref{trans_prob3}, we obtain problem \eqref{trans_prob2}. In the following, we propose an efficient iterative algorithm to solve problem \eqref{trans_prob3} by combining the CCCP and BCD methods, where first-order approximations are employed to convexify the non-convex parts therein and the optimization variables are properly partitioned into different blocks to be successively updated. Note that in contrast, the original $\mathcal{P}_S(\bm{\theta}, \mathcal{H})$ is non-convex in the design variables $\mathbf{w}$ and $\mathbf{P}$ (even with fixed $\mathbf{w}$ or $\mathbf{P}$), which makes it difficult to solve.

First, we note that the SINR constraint \eqref{SINR_const} and the energy harvesting constraint \eqref{EH_constraint} are non-convex and difficult to handle. Therefore, by employing the CCCP method, we introduce the following upper and lower bounds for $\textrm{SINR}_m^{\textrm{E}}(\mathbf{w}, \mathbf{P}, \bm{\theta})$ and $Q_m (\mathbf{w}, \mathbf{P}, \bm{\theta})$, respectively:
\begin{equation}\label{ineq_snr}
	\begin{aligned}
	\textrm{SINR}_m^{\textrm{E}} & (\mathbf{w}, \mathbf{P}, \bm{\theta})  \\ 
	& \leq \frac{|\tilde{\mathbf{g}}_m^H\mathbf{w}|^2}{-\|\tilde{\mathbf{g}}_{m}^H\mathbf{P}^{f}\|^2 + 2\Re\{\tilde{\mathbf{g}}_{m}^H\mathbf{P}^{f}\mathbf{P}^H\tilde{\mathbf{g}}_{m} \}+\sigma_{E,m}^2},
	\end{aligned}
\end{equation}
\begin{equation}\label{ineq_eh}
	\begin{aligned}
	Q_m (\mathbf{w}, \mathbf{P}, \bm{\theta})  \geq & -|\tilde{\mathbf{g}}_{m}^H\mathbf{w}^{f}|^2-\|\tilde{\mathbf{g}}_{m}^H\mathbf{P}^{f}\|^2 \\
	& + 2\Re\{\tilde{\mathbf{g}}_{m}^H\mathbf{w}^{f}\mathbf{w}^H\tilde{\mathbf{g}}_{m} \!+\! \tilde{\mathbf{g}}_{m}^H\mathbf{P}^{f}\mathbf{P}^H\tilde{\mathbf{g}}_{m} \},
	\end{aligned}
\end{equation}
where  $\|\tilde{\mathbf{g}}_{m}^H\mathbf{P}\|^2$ and $|\tilde{\mathbf{g}}_{m}^H\mathbf{w}|^2$ are approximated by their first-order Taylor expansions, and $\mathbf{w}^{f}$ and $\mathbf{P}^{f}$ denote the information beamforming vector and energy beamforming matrix obtained in the previous iteration, which are fixed in the current iteration. Accordingly, we obtain an approximate problem of \eqref{trans_prob3}, shown in \eqref{trans_prob4} at the top of this page.
\begin{figure*}
\begin{subequations} \label{trans_prob4}
\begin{align}
	\min_{\mathbf{w},\mathbf{P}, z,  u, v,\{y_m\}} & \;  ze(\mathbf{w},\mathbf{P},u)-\log(z)+\frac{1}{p}v\left(\sum_{m \in \mathcal{M}}(1+y_m)^{p}\right)-\frac{1}{p}\log(v) \label{obj}\\
	\text{s.t.}\; &  \frac{|\tilde{\mathbf{g}}_m^H\mathbf{w}|^2}{-\|\tilde{\mathbf{g}}_{m}^H\mathbf{P}^{f}\|^2 + 2\Re\{\tilde{\mathbf{g}}_{m}^H\mathbf{P}^{f}\mathbf{P}^H\tilde{\mathbf{g}}_{m} \}+\sigma_{E,m}^2}\leq y_m,\;\forall m \in \mathcal{M}, \label{SINR_constraint1} \\
	&- \! |\tilde{\mathbf{g}}_{m}^H\mathbf{w}^{f}|^2 \! -\! \|\tilde{\mathbf{g}}_{m}^H\mathbf{P}^{f}\|^2 \! + \! 2\Re\{\tilde{\mathbf{g}}_{m}^H\mathbf{w}^{f}\mathbf{w}^H\tilde{\mathbf{g}}_{m} \! + \! \tilde{\mathbf{g}}_{m}^H\mathbf{P}^{f}\mathbf{P}^H\tilde{\mathbf{g}}_{m} \}\geq \epsilon_m,\;\forall m \in \mathcal{M},\label{EH_constraint1}\\
	&\eqref{power_constraint}. \notag
\end{align}
\hrulefill
\end{subequations}
\end{figure*}

To this end, we propose to solve problem \eqref{trans_prob4} in a BCD-fashion, i.e., the variables are divided into multiple blocks and are updated sequentially. Specifically, the variables are divided into three blocks, i.e., $\{z,v\}$, $u$, and $\{\mathbf{w}, \mathbf{P}, y_m\}$, and the detailed updating steps are given as follows.

\textbf{Step 1}: With fixed $u$ and $\{\mathbf{w}, \mathbf{P}, y_m\}$, it can be seen that optimizing $z$ and $v$ is equivalent to solving the following unconstrained convex optimization problem:
\begin{equation} \label{step1_problem}
	\begin{aligned}
		\min_{z,\; v} \; & ze(\mathbf{w},\mathbf{P},u)-\log(z)\\
		& +\frac{1}{p}v\left(\sum_{m \in \mathcal{M}}(1+y_m)^{p}\right)-\frac{1}{p}\log(v),
	\end{aligned}
\end{equation}
whose optimal solution can be easily obtained as $z = \frac{1}{e(\mathbf{w},\mathbf{P},u)}$ and $v = \frac{1}{\sum_{m \in \mathcal{M}}(1+y_m)^p}$ by resorting to the first-order optimality condition.

\textbf{Step 2}: With fixed $\{\mathbf{w}, \mathbf{P}, y_m\}$ and  $\{z, v\}$, we obtain the following convex problem with respect to (w.r.t.) $u$:
\begin{equation} \label{step2_problem}
	\begin{aligned}
\min\limits_{u}\;& |u|^2| \tilde{\mathbf{h}}^H\mathbf{w}|^2+|u|^2\sigma^2\\
& +|u|^2\|\tilde{\mathbf{h}}^H\mathbf{P}\|^2+1-2\Re\{u^* \tilde{\mathbf{h}}^H\mathbf{w}\},
\end{aligned}
\end{equation}
and the optimal solution is given by $u = \frac{\tilde{\mathbf{h}}^H\mathbf{w}}{|\tilde{\mathbf{h}}^H\mathbf{w}|^2+\|\tilde{\mathbf{h}}^H\mathbf{P}\|^2+\sigma^2}$.

\textbf{Step 3}: For the optimization of  $\{\mathbf{w}, \mathbf{P}, y_m\}$, we have the following problem:
\begin{equation} \label{step3_problem}
\begin{aligned}
	\min_{\mathbf{w},\; \mathbf{P},\; \{y_m\}} \; & ze(\mathbf{w},\mathbf{P},u)-\log(z)\\
	& +\frac{1}{p}v\left(\sum_{m \in \mathcal{M}}(1+y_m)^p\right)-\frac{1}{p}\log(v)  \notag \\
	\text{s.t.}\; & \eqref{power_constraint}, \eqref{SINR_constraint1}\;\textrm{and}\; \eqref{EH_constraint1},
\end{aligned}
\end{equation}
which is convex and can be efficiently solved via off-the-shelf solvers, such as CVX \cite{CVX}.

To summarize, the short-term subproblem $\mathcal{P}_S(\bm{\theta}, \mathcal{H})$ can be efficiently solved by iterating over the
abovementioned three steps and the details are shown in Algorithm \ref{inner_algorithm}. With regard to the convergence property of Algorithm \ref{inner_algorithm}, we have the following theorem.
\begin{theorem} \label{theorem_convergence}
	The objective value \eqref{obj} is non-increasing over the iterations of Algorithm \ref{inner_algorithm} and any limit point generated by Algorithm \ref{inner_algorithm} is a stationary solution of $\mathcal{P}_S(\bm{\theta}, \mathcal{H})$.
\end{theorem}
\begin{proof}
	Please refer to Appendix \ref{proof_of_convergence}.
\end{proof}

\begin{algorithm}
\caption{Proposed CCCP-BCD algorithm for solving $\mathcal{P}_S(\bm{\theta}, \mathcal{H})$}
\begin{algorithmic}[1] \label{inner_algorithm}
\STATE Initialize the variables $\mathbf{w}^0,\mathbf{P}^0, z^0, u^0, v^0,\{y_m^0\}$ with a feasible point. Set the iteration index $i=1$.
\REPEAT
		\STATE Update $\{z,v\}$, $u$ and $\{\mathbf{w}, \mathbf{P}, y_m\}$ successively by solving problems \eqref{step1_problem}, \eqref{step2_problem} and \eqref{step3_problem}, respectively, and obtain $\{\mathbf{w}^i,\mathbf{P}^i, z^i, u^i, v^i,y_m^i\}$.
		\STATE Set $\mathbf{P}^{f}=\mathbf{P}^i$ and $\mathbf{w}^{f}=\mathbf{w}^i$, and let $i \leftarrow i+1$.
\UNTIL{The maximum number of $J$ iterations is reached.}	
	\end{algorithmic}
\end{algorithm}

\subsection{Solving the Long-term Subproblem}
In this subsection, we propose to solve the long-term problem $\mathcal{P}_L^J$ based on the SSCA framework proposed in \cite{Liu2018TOSCA}, where an iterative algorithm is presented and in each iteration (or equivalently, at the end of each frame), a convex surrogate problem is constructed to resolve the difficulty that no closed-form expression of  $\mathbb{E}\{\bar{S} (\mathbf{w}^J(\bm{\theta},\mathcal{H}), \mathbf{P}^J(\bm{\theta},\mathcal{H}), \bm{\theta},\mathcal{H}) \}$ is available.\footnote{\color{black}Note that the objective function of problem \eqref{longterm_problem} is a complicated function, which depends on both the information and energy beamforming vectors, as well as on the statistical and effective channel information contained in $\mathcal{H}$; therefore, directly maximizing it through optimizing the RIS phase shifts is intractable.} We show that the surrogate problem can be simply constructed by replacing the objective function of $\mathcal{P}_L^J$ with a concave quadratic function which depends on the channel samples obtained at the current and preceding iterations (frames). Besides, the global optimal solution of the resulting surrogate problem can be expressed in a simple closed-form and the long-term RIS phase shifts are iteratively updated.


Specifically, according to \cite{Liu2018TOSCA}, $\mathcal{P}_L^J$ can be approximated as follows at the $t$-th iteration:
\begin{equation}\label{approximation_problem}
	\begin{aligned}
	\max_{\boldsymbol{\theta}} \; & \bar{g}^t(\boldsymbol{\theta})\\
\textrm{s.t.}\; & 0\leq\theta_n < 2\pi, \forall n \in \mathcal{N},
	\end{aligned}
\end{equation}
where  $\bar{g}^t(\boldsymbol{\theta})$ is the surrogate objective function that is given by
\begin{equation}
	\bar{g}^t(\boldsymbol{\theta})=f^t+(\mathbf{f}^t)^T(\boldsymbol{\theta}-\boldsymbol{\theta}^t) - \tau\|\boldsymbol{\theta}-\boldsymbol{\theta}^t\|^2,\label{surrogate_function}
\end{equation}
which can be viewed as a concave approximation of the objective function $\mathbb{E}\{\bar{S} (\mathbf{w}^J(\bm{\theta},\mathcal{H}), $ $\mathbf{P}^J(\bm{\theta},\mathcal{H}), \bm{\theta}) \}$ of problem \eqref{longterm_problem}. $ \tau> 0$ can be any constant. $f^t$ and $\mathbf{f}^t$ are the approximations of the objective function value $\mathbb{E}\{\bar{S} (\mathbf{w}^J(\bm{\theta}^t,\mathcal{H}), \mathbf{P}^J(\bm{\theta}^t,\mathcal{H}), \bm{\theta}^t) \}$ and the partial derivative of  $\mathbb{E}\{\bar{S} (\mathbf{w}^J(\bm{\theta},\mathcal{H}), $ $\mathbf{P}^J(\bm{\theta},\mathcal{H}), \bm{\theta}) \}$ w.r.t. $\bm{\theta}$ at $\bm{\theta} = \bm{\theta}^t$, respectively, which are iteratively updated as
\begin{equation} \label{gradient}
	\begin{aligned}
	f^{t}=& (1-\rho^t)f^{t-1}\\
	&+\frac{\rho^t}{T_c}\sum_{j=1}^{T_c} \bar{S}(\mathbf{w}^J(\bm{\theta}^t,\mathcal{H}_j^t), \mathbf{P}^J(\bm{\theta}^t,\mathcal{H}_j^t), \bm{\theta}^t,\mathcal{H}_j^t),\\
	\mathbf{f}^t=& (1-\rho^t)\mathbf{f}^{t-1}\\
	& +\frac{\rho^t}{T_c}\sum_{j=1}^{T_c}\frac{\partial \bar{S}}{\partial \boldsymbol{\theta}}\bigg|_{(\mathbf{w}^J(\bm{\theta}^t,\mathcal{H}_j^t), \mathbf{P}^J(\bm{\theta}^t,\mathcal{H}_j^t),\boldsymbol{\theta}^t, \mathcal{H}^t_j)},
	\end{aligned}
\end{equation}
with $f^{-1} = 0$ and $\mathbf{f}^{-1} = \mathbf{0}$, where $\mathcal{H}^t_j\triangleq\{\mathbf{h}_{1,j}^t, \mathbf{h}_{2,j}^t, \mathbf{F}_{1,j}^t$, $\{\mathbf{g}_{1,m,j}^t, \mathbf{g}_{2,m,j}^t\} \}, j\in\{1,\cdots,T_c\}$ denote the channel samples available at the $t$-th iteration, and the detailed derivation of $\frac{\partial \bar{S}}{\partial \boldsymbol{\theta}}\big|_{(\mathbf{w}^J(\bm{\theta}^t,\mathcal{H}_j^t), }$ $_{ \mathbf{P}^J(\bm{\theta}^t,\mathcal{H}_j^t),\boldsymbol{\theta}^t, \mathcal{H}^t_j)}$ is given in Appendix \ref{derivative}.  $\{\rho^t\}$ is a sequence of forgetting factors that satisfies $\rho^t\to 0$, $\frac{1}{\rho^t}\leq \mathcal{O}(t^{\beta})\text{ for some }\beta\in(0,1)$, and $\sum_t (\rho^t)^2<\infty$.
{\color{black} We emphasize that $\mathcal{H}^t_j, j\in\{1,\cdots,T_c\}$ are not the instantaneous channels, but are channel samples generated from the statistical CSI.} Note that as $t\rightarrow \infty $, $\bar{g}^t(\boldsymbol{\theta})$ exhibits the following asymptotic consistency properties \cite{Liu2018TOSCA}:
\begin{equation}
	\begin{aligned}
	\lim \limits_{t\rightarrow \infty} |\bar{g}^t(\boldsymbol{\theta}^t) \!-\! \mathbb{E}\{\bar{S} (\mathbf{w}^J(\bm{\theta}^t,\mathcal{H}), \mathbf{P}^J(\bm{\theta}^t,\mathcal{H}), \bm{\theta}^t) \}| = 0,\\
		\lim \limits_{t\rightarrow \infty} \| \nabla \bar{g}^t(\boldsymbol{\theta}^t) \!-\!  \nabla_{\bm{\theta}}\mathbb{E}\{\bar{S} (\mathbf{w}^J(\bm{\theta}^t,\mathcal{H}), \mathbf{P}^J(\bm{\theta}^t,\mathcal{H}), \bm{\theta}^t) \}\| = 0,
	\end{aligned}
\end{equation}
which show that the iterative approximation $\bar{g}^t(\boldsymbol{\theta})$ and its gradient can converge to the true objective function $\mathbb{E}\{\bar{S} (\mathbf{w}^J(\bm{\theta},\mathcal{H}), \mathbf{P}^J(\bm{\theta},\mathcal{H}), \bm{\theta}) \}$ and the corresponding gradient w.r.t. $\bm{\theta}$, as $t\rightarrow \infty$.

Then, the optimal solution of problem \eqref{approximation_problem} can be obtained in closed-form, which is shown as follows:
\begin{equation}
	\bar{\boldsymbol{\theta}}^t =\boldsymbol{\theta}^t+\frac{1}{2\tau}\mathbf{f}^t.
\end{equation}
Accordingly, the long-term variable $\boldsymbol{\theta}$ can be updated as
\begin{equation}
	\boldsymbol{\theta}^{t+1}=(1-\gamma^t)\boldsymbol{\theta}^t+\gamma^{t}\bar{\boldsymbol{\theta}}^t,\label{update_rule}
\end{equation}
where $\{\gamma^t\}$ is a sequence that satisfy $\gamma^t\to 0$, $\sum_t \gamma^t=\infty$, $\sum_t(\gamma^t)^2<\infty$ and $\lim_{t\to \infty} \frac{\gamma^t}{\rho^t}=0$. {\color{black}Since $\gamma^t$ satisfies $\gamma^t \in (0,1]$ and $\bar{\boldsymbol{\theta}}^t \in (0, 2\pi]^N$, the reflection phase shifts $\boldsymbol{\theta}^{t}, \forall t$, will  automatically lie in the feasible region, i.e., $(0, 2\pi]^N$, according to the updating rule in \eqref{update_rule}. } 

{\color{black}Note that for the more practical case of discrete phase shifts at the RIS, we can simply project the entries of $\boldsymbol{\theta}^{t}$  independently onto $\mathcal{F}_d \triangleq \{0,\frac{2\pi}{L}, \cdots, \frac{2\pi(L-1)}{L} \}$ ($L=2^Q$ and $Q$ denotes the number of control bits for phase-shifting per RIS element) to obtain a unit-modulus discrete solution, i.e.,
	\begin{equation}
		\hat{\theta}_n^t = \arg\min\limits_{\hat{\theta}_n^t \in \mathcal{F}_d}\; |\angle \theta_n^t - \angle \hat{\theta}_n^t|, \;\forall n \in \mathcal{N}.
	\end{equation} 
 It will be shown in Section \ref{Section_Simulation} that when $Q \geq 2 $ bits, the performance loss due to discrete phase shifts is negligible.
}

\subsection{Overall Algorithm and Convergence/Complexity Analysis} \label{convergence_complexity}
To summarize, the overall algorithm for solving problem \eqref{trans_prob1} is given in Algorithm \ref{outter_algorithm}. Besides,
Algorithm \ref{outter_algorithm} almost surely converges to the set of stationary solutions of problem \eqref{trans_prob1} and the detailed proof can be found in \cite{Liu2018TOSCA}. Due to Lemma \ref{lemma1}, the obtained solution is also the approximate stationary solution of problem \eqref{orignal_prob}.

From the above, it is observed that the complexity of Algorithm \ref{outter_algorithm} is mainly due to solving $\mathcal{P}_S(\bm{\theta}, \mathcal{H})$ in each time slot and computing $\{\mathbf{w}^J(\bm{\theta}^t,\mathcal{H}_j^t), \mathbf{P}^J(\bm{\theta}^t,\mathcal{H}_j^t)\}$ for the generated channel samples in the long-term optimization problem, which both involve solving problem \eqref{step3_problem} multiple times. By applying the basic elements of complexity analysis as used in \cite{Wang2014}, the complexity of solving problem \eqref{step3_problem} is $\mathcal{O}((N_sM)^{3.5})$. {\color{black}Furthermore, the complexity of computing the partial derivatives (as shown in Appendix \ref{derivative}) is $\mathcal{O}(N_s^2 M^2 + MN_s N_r)$.} Therefore, the overall complexity of Algorithm \ref{outter_algorithm} is given by $\mathcal{O}(T_f T_s (N_sM)^{3.5} + T_f (T_c (N_sM)^{3.5} +T_c M N_s N_r) )$.


\begin{algorithm}
\caption{Proposed SA-SSCA algorithm for solving problem \eqref{trans_prob1}}
\begin{algorithmic} [1] \label{outter_algorithm}
\STATE At the beginning of each super-frame, initialize the long-term variable $\boldsymbol{\theta}^0$ and the short-term variables $(\mathbf{w}^0, \mathbf{P}^0)$. Set the frame index $t=0$ and the time slot index $i=0$.
\REPEAT
\REPEAT
\STATE Obtain the effective CSI $\{\tilde{\mathbf{h}},\{\tilde{\mathbf{g}}_m\}\}$ with given $\bm{\theta}^t$ in time slot $i$.
\STATE Solve problem $\mathcal{P}_S(\bm{\theta}, \mathcal{H})$ using Algorithm \ref{inner_algorithm} and obtain the solution $\{\mathbf{P}^i, \mathbf{w}^i\}$.
\STATE $i\leftarrow i+1$.
\UNTIL{the frame ends, i.e., $i=(t+1)T_s - 1$.	}
\STATE Obtain $T_c$ channel samples $\{\mathcal{H}^t_j\}_{j=1}^{T_c}$ at the end of frame $t$.
\STATE Update the surrogate function \eqref{surrogate_function} using $\{\mathbf{w}^J(\bm{\theta}^t,\mathcal{H}_j^t), \mathbf{P}^J(\bm{\theta}^t,\mathcal{H}_j^t)\}$ (obtained via running Algorithm \ref{inner_algorithm}) and $\{\mathcal{H}^t_j\}_{j=1}^{T_c}$.
\STATE Solve problem \eqref{approximation_problem} to obtain $\bar{\boldsymbol{\theta}}^t$ and update $\boldsymbol{\theta}^{t+1}$ according to \eqref{update_rule}.
\STATE $t\leftarrow t+1$ .
\UNTIL{the super-frame ends, i.e., $t = T_f-1$.}
\end{algorithmic}
\end{algorithm}

%

\section{Low-Complexity Algorithm} \label{section_low_complexity}
In this section, in order to reduce the computational complexity of the proposed SA-SSCA algorithm, we present a low-complexity efficient heuristic
algorithm to solve the long-term subproblem \eqref{longterm_problem}. The proposed low-complexity algorithm is motivated by the fact that for each channel sample $\mathcal{H}^t_j$ available at the end of frame $t$, the proposed CCCP-BCD algorithm in Algorithm \ref{inner_algorithm} needs to be implemented to obtain the corresponding short-term transmit beamforming vectors at the BS, i.e., $\{\mathbf{w}^J(\bm{\theta}^t,\mathcal{H}_j^t), \mathbf{P}^J(\bm{\theta}^t,\mathcal{H}_j^t)\}$, and this step (Step 9 in Algorithm \ref{outter_algorithm}) is required to be repeatedly conducted at the end of each frame, which incurs high complexity.

To address this difficulty, we design the long-term variable $\bm{\theta}$ based on the intuition that in order to maximize the worst-case secrecy rate, we need to increase the effective channel power from the BS to the IU, i.e., $\|\tilde{\mathbf{h}}\|^2$, and in the meantime reduce the effective channel power from the BS to the EUs, i.e., $\|\tilde{\mathbf{g}}_m\|^2, \forall m$. Therefore, the proposed objective function for low-complexity long-term optimization can be formulated as follows:
\begin{equation} \label{low_complexity_metric}
	C_L(\bm{\phi}) = \mathbb{E}\left\{-\|\tilde{\mathbf{h}}\|^2 + a\sum_{m=1}^M\|\tilde{\mathbf{g}}_m\|^2\right\},
\end{equation}
where $a$ is the weighting factor that is chosen carefully to balance between desired channel power enhancement and eavesdropping channel power suppression, and the expectation is taken over the channel statistics. As a result, the lone-term subproblem can be approximated as
\begin{equation}
\begin{aligned}
	\min_{\bm{\phi}} \; &C_L(\boldsymbol{\phi})\\
	\textrm{s.t.}\; & |\phi_n|=1, \forall n \in \mathcal{N}. \label{low_prob}
\end{aligned}
\end{equation}
To proceed, we rewrite the objective function $C_L(\bm{\phi})$ as follows:
\begin{equation}
	C_L(\bar{\bm{\phi}}) \!=\! {\bar{\bm{\phi}}}^H \mathbb{E}\left\{-{\bar{\mathbf{H}}}^H{\bar{\mathbf{H}}} \!+\! a\sum_{m=1}^M{\bar{\mathbf{G}}}_m^H{\bar{\mathbf{G}}}_m\right\}{\bar{\bm{\phi}}} \!=\! {\bar{\bm{\phi}}}^H \bar{\mathbf{A}} {\bar{\bm{\phi}}},
\end{equation}
where ${\bar{\mathbf{H}}} = [\mathbf{F}_1\text{diag}(\mathbf{h}_2), \mathbf{h}_1]$, ${\bar{\mathbf{G}}}_m = [\mathbf{F}_1\text{diag}(\mathbf{g}_{2, m}), \mathbf{g}_{1,m}]$, ${\bar{\bm{\phi}}}=[\boldsymbol{\phi}^T, 1]^T$ and ${\bar{\mathbf{A}}} = \mathbb{E}\{-{\bar{\mathbf{H}}}^H{\bar{\mathbf{H}}}+a\sum_{m=1}^M{\bar{\mathbf{G}}}_m^H{\bar{\mathbf{G}}}_m\}$. Note that the matrix ${\bar{\mathbf{A}}}$ depends on the channel statistics, which can be easily obtained via sample averaging.

Then, it is observed that $C_L(\bar{\bm{\phi}})$ is a quadratic function w.r.t. $\bar{\bm{\phi}}$, which, however, may not be convex since $\bar{\mathbf{A}}$ is not always positive semidefinite. Therefore, problem \eqref{low_prob} is still a non-convex problem with uni-modular constraints. In the literature, various methods have been proposed to address a similar problem by using, e.g., the SSCA method as introduced above, the semidefinite relaxation (SDR) method \cite{Wu2018_journal}, the BCD method, and the penalty dual decomposition (PDD) method \cite{zhao2019intelligent}. In this work, we employ the last one, since it enables the optimization of RIS phase shifts in parallel and its computational complexity is low. Interested readers may refer to \cite{zhao2019intelligent} for the details. Besides, it is noteworthy that for the proposed low-complexity algorithm, the weighting factor $a$ should be carefully chosen to achieve good performance. In principle, a one-dimensional search should be conducted to find the value of $a$, which, however, may result in high computational complexity. To address this problem, we propose a simple method to determine the value of $a$ that can  empirically achieve good performance in our numerical simulations, which is shown as follows:
\begin{equation}
	a = \frac{\mathbb{E}\{\|\mathbf{h}_1\|^2\}}{\mathbb{E}\left\{\frac{1}{M}\sum\limits_{m=1}^M\|\mathbf{g}_{1,m}\|^2\right\}}.
\end{equation}
Note that the intuition here is to balance the BS-IU and BS-EU effective channel powers approximately. Finally, similar to the complexity analysis in Section \ref{convergence_complexity}, the complexity of the proposed low-complexity algorithm can be expressed as $\mathcal{O}(T_f T_s (N_sM)^{3.5} + I N_r^2 )$, where $I$ denotes the number of iterations required to solve problem \eqref{low_prob} (including both outer and inner iterations).

%

\section{Simulation Results} \label{Section_Simulation}
In this section, we provide numerical results by simulations to verify the effectiveness of the proposed algorithms and draw useful insights. We assume that  the RIS is equipped with $ N_r = 80$ reflecting elements without loss of generality. The path losses of $\mathbf{F}_1$, $\mathbf{h}_1$, $\mathbf{h}_2$,
$\mathbf{g}_{1,m}$'s and  $\mathbf{g}_{2,m}$'s are modeled as $L^{\textrm{BR}} = C_0\left({d^{\textrm{BR}}}/{D_0}\right)^{-\alpha^{\textrm{BR}}}$, $L^{\textrm{BI}} = C_0\left({d^{\textrm{BI}}}/{D_0}\right)^{-\alpha^{\textrm{BI}}}$, $L^{\textrm{RI}} = C_0\left({d^{\textrm{RI}}}/{D_0}\right)^{-\alpha^{\textrm{RI}}}$, $L_k^{\textrm{BE}} = C_0\left({d_k^{\textrm{BE}}}/{D_0}\right)^{-\alpha_k^{\textrm{BE}}}$ and $L_k^{\textrm{RE}} = C_0\left({d_k^{\textrm{RE}}}/{D_0}\right)^{-\alpha_k^{\textrm{RE}}}$, where $D_0 = 1$ meter (m) denotes the reference distance, $C_0 = -30$ dB denotes the path loss at the reference distance,  ${d^{\textrm{BR}}}$, ${d^{\textrm{BI}}}$, ${d^{\textrm{RI}}}$, ${d_k^{\textrm{BE}}}$ and ${d_k^{\textrm{RE}}}$ denote the link distance from the BS to the RIS, from the BS to the IU, from the RIS to the IU, from the BS to EU $m$, and from the RIS to EU $m$, respectively, while $\alpha^{\textrm{BR}}$, $\alpha^{\textrm{BI}}$, $\alpha^{\textrm{RI}}$, $\alpha_k^{\textrm{BE}}$, and $\alpha_k^{\textrm{RE}}$ denote the corresponding path loss factors. We assume that the RIS is deployed to serve the users that suffer from severe signal attenuation in the BS-user direct links and thus we set $\alpha^{\textrm{BI}} = \alpha_k^{\textrm{BE}} = 3.6$ and $\alpha^{\textrm{BR}} = \alpha^{\textrm{RI}} = \alpha_k^{\textrm{RE}} = 2.2$, i.e., the path-loss exponent of the BS-user link is larger than those of the BS-RIS and RIS-user links.
Besides, in our simulations, a three-dimensional coordinate system is considered, as shown in Fig. \ref{fig:usersetup}, where the BS (equipped with a uniform linear array (ULA)) and the RIS (equipped with a uniform planar array (UPA)) are located on the $x$-axis and $y$-$z$ plane, respectively.
The reference antenna/element at the BS/RIS are located at $(6\;\textrm{m}, 0, 0)$ and $(0\;\textrm{m}, 2.5\;\textrm{m}, 3\;\textrm{m})$, the location of the IU is $(6\;\textrm{m}, 200\;\textrm{m}, 0)$ and the EUs are located in a circle centering at the AP with a radius of $5$ m.
To account for small-scale fading, we assume the general Rician fading channel model for all channels involved and the details are similar to those in \cite{zhao2019intelligent}. The Rician factors of the BS-user links are set to $0$ dB, while those of the BS-RIS and RIS-user links are set to $3$ dB, without loss of generality.  Other system parameters are set as follows unless otherwise specified:  $\sigma_{I}^2 = \sigma_{E,m}^2 = -80$ dBm, $\forall m$, $P_t=45$ dBm, $M=6$, $N_s = 2$, $\epsilon_m =\epsilon= 2\; \mu $w, $\forall m$, $p=4$. All the results are averaged over 1000 independent channel realizations.

To illustrate the performance gain of our proposed algorithms, we consider the following benchmark schemes.
\begin{itemize}
	\item Instantaneous CSI based scheme: the transmit beamforming vectors at the BS and the RIS phase shifts are iteratively and alternatively optimized by employing the BCD method and assuming perfect instantaneous CSI over all time slots.
	\item BS-IU effective channel power maximization scheme: the RIS phase shifts are designed by solving problem \eqref{low_prob} with fixed $a=0$ (i.e., maximizing the BS-IU effective channel power), and the active beamforming vectors are designed according to Algorithm \ref{inner_algorithm} with fixed RIS phase shifts.
	\item Random phase-shift scheme:  the phase shifts at the RIS are randomly generated at each time slot, and the active beamforming vectors are designed according to Algorithm \ref{inner_algorithm}.
\end{itemize}

\begin{figure}
\centering
	\setlength{\abovecaptionskip}{-0cm}
\setlength{\belowcaptionskip}{-0cm}
\scalebox{0.58}{\includegraphics{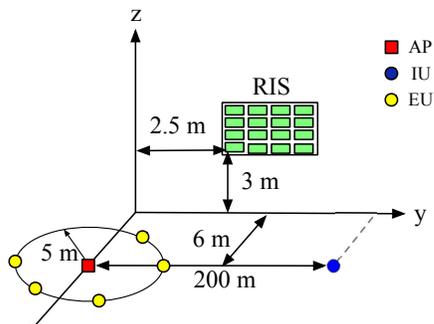}}
\caption{Simulation setup.}\label{fig:usersetup}
\end{figure}

{\color{black}
\subsection{Impact of Batch Size, $T_c$}
First, in Fig. \ref{fig_compare_convergence}, we investigate the convergence behavior of the proposed SA-SSCA algorithm for different values of $T_c$. Note that $T_c$ can be regarded as the batch size for the SA-SSCA algorithm, as it represents the number of channel/training samples utilized in one iteration. As seen, the choice of $T_c$ does not affect the convergence speed and steady state performance much and the proposed algorithm converges for all values of $T_c$ tested in $\{1,2,4,8,16\}$. However, the curves with larger $T_c$ are more smooth, which is expected since the partial derivative $\mathbf{f}^t$ (defined in \eqref{gradient}) extracted from a larger number of channel samples is generally more accurate than that with a smaller $T_c$. 
\begin{figure}[t]
	\centering
	\setlength{\abovecaptionskip}{-0cm}
	\setlength{\belowcaptionskip}{-0cm}
	\scalebox{0.48}{\includegraphics{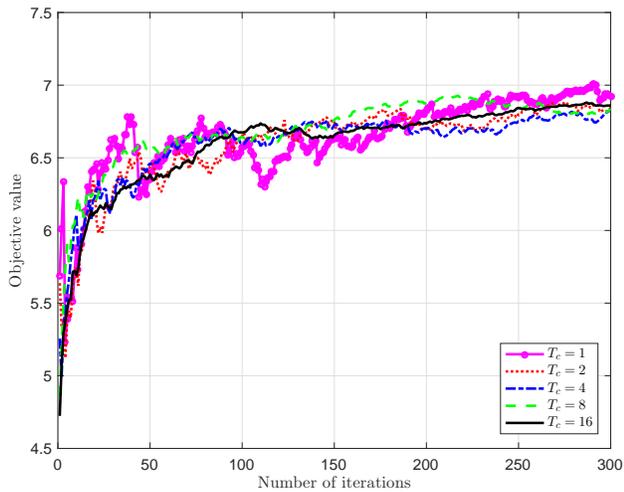}}
	\caption{Convergence behavior of the proposed SA-SSCA algorithm with different values of $T_c$.}\label{fig_compare_convergence}
\end{figure}
}

\subsection{Impact of BS Transmit Power Budget, $P_t$}
\begin{figure}[t]
	\centering
		\setlength{\abovecaptionskip}{-0cm}
	\setlength{\belowcaptionskip}{-0cm}
	\scalebox{0.48}{\includegraphics{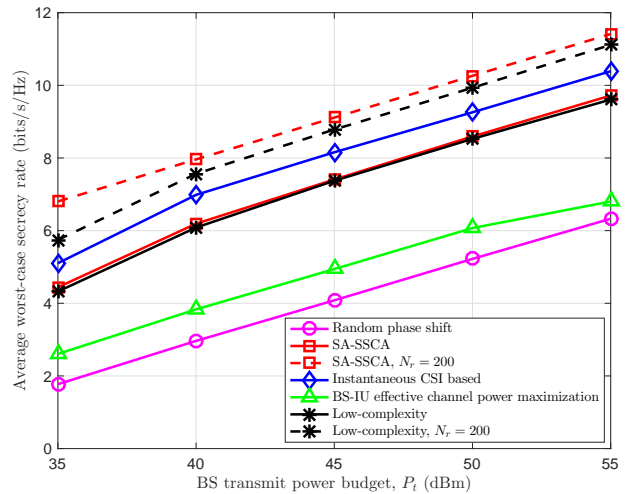}}
	\caption{Average worst-case secrecy rate versus BS transmit power budget, $P_t$.}\label{fig1}
\end{figure}
Second, we investigate the impact of the BS transmit power budget, i.e., $P_t$,  on the system performance. In this experiment, we fix all the other system parameters as mentioned above and change $P_t$ from $35$ dBm to $55$ dBm. As shown in Fig. \ref{fig1}, the performance of all the schemes improves with the increasing of $P_t$ and {\color{black}the instantaneous CSI based scheme achieves the best performance, followed by the proposed two-timescale schemes}, i.e., the SA-SSCA algorithm introduced in Section \ref{sec_proposed_algorithm}, the low-complexity algorithm in Section \ref{section_low_complexity}, the BS-IU effective channel power maximization scheme, and the random phase-shift scheme.
Besides, it can be seen that the performance of the proposed low-complexity algorithm is very close (within 0.1 bits/s/Hz) to that of the SA-SSCA algorithm when $N_r=80$, which verifies the effectiveness of the proposed metric in \eqref{low_complexity_metric} designed for updating the long-term phase shifts. {\color{black}Their performance gap enlarges when $N_r$ is larger (the performance when $N_r=200$ is shown in Fig. \ref{fig1}), which is also reasonable since typically more sophisticated reflection control is required for a larger RIS and  the strategy adopted by the low-complexity algorithm is relatively simple.}

\subsection{Impact of Number of EUs, $M$}
Then, we investigate the performance of the considered schemes under various numbers of EUs, as shown in Fig. \ref{fig2}. It is observed that the performance of all the schemes gradually deteriorates with more EUs. This is expected since as $M$ increases, the probability that the maximum achievable rate of the EUs becomes larger increases. Besides, due to the more energy harvesting constraints when $M$ increases, less power will be allocated for information transmission, which also leads to lower secrecy rate. {\color{black} By comparing the performance of the schemes considered, we can see that  similar trends to those in Fig.~\ref{fig1} can be observed. In particular, the performance gap between the instantaneous CSI based scheme and the proposed SA-SSCA and low-complexity algorithms slightly expands as $M$ increases}, which is reasonable since allowing the reflection phase shifts to be tuned in each time slot provides more flexibility for improving the secrecy and SWIPT performance.


\begin{figure}[t]
	\centering
		\setlength{\abovecaptionskip}{-0cm}
	\setlength{\belowcaptionskip}{-0cm}
	\scalebox{0.48}{\includegraphics{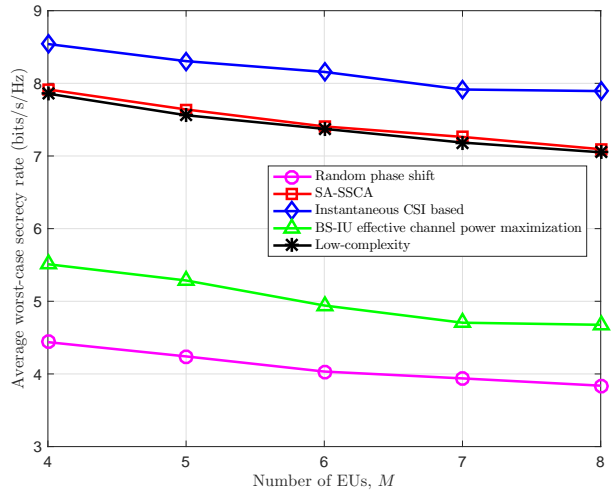}}
	\caption{Average worst-case secrecy rate versus number of EUs, $M$.}\label{fig2}
\end{figure}

\subsection{Impact of Minimum RF Receive Power Requirement, $\epsilon$}
Next, we study in Fig. \ref{fig3} the impacts of the RF receive power requirement $\epsilon$ on the secrecy performance, where $M$ is set to 4 and the other system parameters are set as default. It is observed that the worst-case secrecy rate performance gradually decreases with the increasing of $\epsilon$, which is expected since the IU and EUs are far apart in our setup and a larger portion of transmit power is allocated to satisfy the energy harvesting requirements of the EUs as $\epsilon$ increases. Besides, we also observe that the worst-case secrecy rate does not drop much (within $0.5$ bits/s/Hz) as $\epsilon$ increases from $4\;\mu$w to $60\;\mu$w. This is due to the fact that the energy harvesting constraints are relatively easier to address as compared to maximizing the worst-case secrecy rate, since the former is only related with the BS-EU and BS-RIS-EU links, while the latter should further take the BS-IU and BS-RIS-IU links into consideration. Therefore, by sacrificing a small amount of worst-case secrecy rate, the harvested energy can be significantly increased.
\begin{figure}[t]
	\centering
		\setlength{\abovecaptionskip}{-0cm}
	\setlength{\belowcaptionskip}{-0cm}
	\scalebox{0.48}{\includegraphics{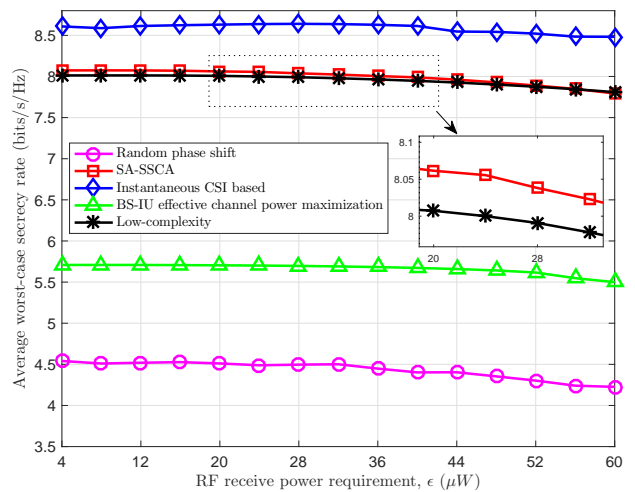}}
	\caption{Average worst-case secrecy rate versus minimum RF receive power requirement, $\epsilon$.}\label{fig3}
\end{figure}

\subsection{Impact of Number of RIS Reflecting Elements, $N_r$}
In Fig. \ref{fig4}, we show the average worst-case secrecy rate performance of the considered schemes under various numbers of reflecting elements, $N_r$.
It can be observed that the performance of all the algorithms, except for the BS-IU effective power maximization scheme, improves with the increasing of $N_r$. This is mainly due to the fact that larger $N_r$ leads to higher aperture gain and in the meantime more fine-grained reflect beamforming, that is able to improve the overall system performance. However, for the BS-IU effective power maximization scheme, the IU's receive signal power gradually gets saturated as $N_r$ increases and since the EUs have much shorter distances to the RIS and AP as compared to the IU, their receive signal power increases more quickly with $N_r$, which in turn decreases the secrecy rate.

\begin{figure}[t]
	\centering
		\setlength{\abovecaptionskip}{-0cm}
	\setlength{\belowcaptionskip}{-0cm}
	\scalebox{0.48}{\includegraphics{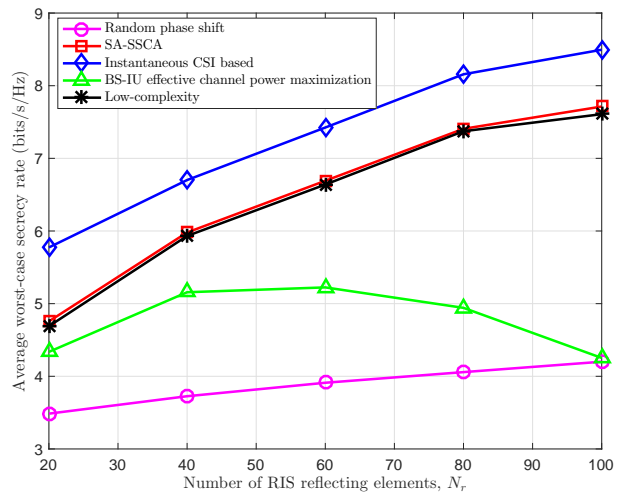}}
	\caption{Average worst-case secrecy rate versus number of RIS reflecting elements, $N_r$.}\label{fig4}
\end{figure}

{\color{black}
\subsection{Impact of Discrete Phase Shifts}
Next, in Fig. \ref{figure_compare_Q} we portray the impact of discrete phase shifts on the performance of the proposed algorithm. It is observed that as the number of control bits per RIS element (i.e., $Q$) increases, the performance gap between the discrete phase-shift and the continuous phase-shift case  gradually decreases and when $Q=3$, using RISs relying on discrete phase shifters incurs only negligible performance erosion. Furthermore, it can be seen that compared to the proposed SA-SSCA algorithm, the instantaneous CSI based scheme is in general more sensitive to $Q$, and it requires about $4$ bits to achieve similar performance to that of the continuous phase-shift scenario.
}
\begin{figure}[t]
	\centering
	\setlength{\abovecaptionskip}{-0cm}
	\setlength{\belowcaptionskip}{-0cm}
	\scalebox{0.48}{\includegraphics{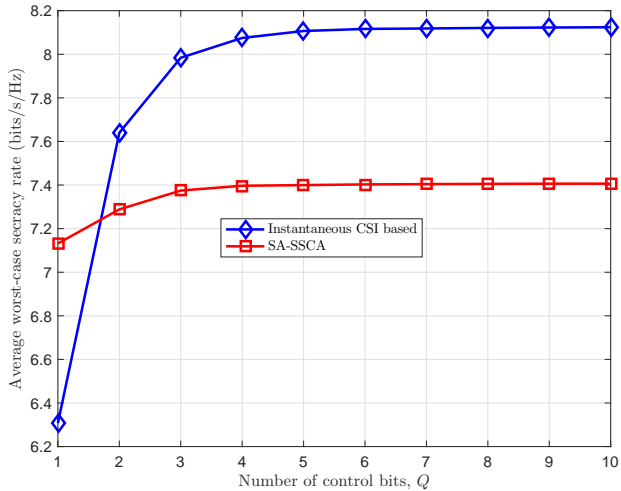}}
	\caption{Average worst-case secrecy rate versus number of quantization bits per RIS element, $Q$.}\label{figure_compare_Q}
\end{figure}

\subsection{Impact of CSI Delay}
Then, we investigate the impact of CSI delay (for BS-user effective channel estimation) on the worst-case secrecy rate performance. In this simulation, we adopt a dynamic CSI model as in \cite{biglieri2007mimo} to model the CSI delay, for example, the BS-RIS channel at time $s$ can be modeled as
\begin{equation}
	\mathbf{F}_1(s) = \rho(s) \mathbf{F}_1(s) + (1-\rho(s)) \mathbf{F}_{1,m},
\end{equation}
where $\mathbf{F}_{1,m}$ is the channel mean, $\rho(s) = J_0(2 \pi f_d s)$, $f_d$ is the channel Doppler spread, and $J_0(\cdot)$ is the zeroth-order Bessel function of the first kind. The delay of the other channels are modeled similarly. We assume that the CSI delay is proportional to the number of  channel coefficients that are required to be estimated, i.e., if the CSI delay for BS-user effective channel estimation (i.e., $\tilde{\mathbf{h}}$ and $\tilde{\mathbf{g}}_m$) is $\omega$ millisecond (ms), then the CSI delay for estimating the full channel sample $\{\mathbf{F}_1, \mathbf{h}_1, \mathbf{h}_2, \mathbf{g}_{1,m}, \mathbf{g}_{2,m}\}$ is $\frac{N_s(M+1)+N_rN_s+ N_r(M+1)}{N_s(M+1)} \omega$ ms. Note that due to the adopted two-timescale-based scheme, the CSI delays in the proposed algorithms are much smaller than that in the instantaneous CSI based scheme. As shown in Fig. \ref{fig5}, the performance of the instantaneous CSI based scheme deteriorates very quickly with the increasing of the CSI delay, while those of the two-timescale-based schemes are almost invariant for different values of $\omega$ (at least in the range $[0\;\textrm{ms} - 20\;\textrm{ms}]$). This result shows the superiority of the proposed two-timescale-based schemes in practice as the channel estimation overhead and implementation cost can both be reduced, while the worst-case secrecy rate performance is not compromised much. {\color{black}It is important to mention that the proposed two-timescale-based schemes still have to estimate a number of full channels for statistical CSI estimation, albeit not instantaneously.}

\begin{figure}[t]
	\centering
	\setlength{\abovecaptionskip}{-0cm}
	\setlength{\belowcaptionskip}{-0cm}
	\scalebox{0.48}{\includegraphics{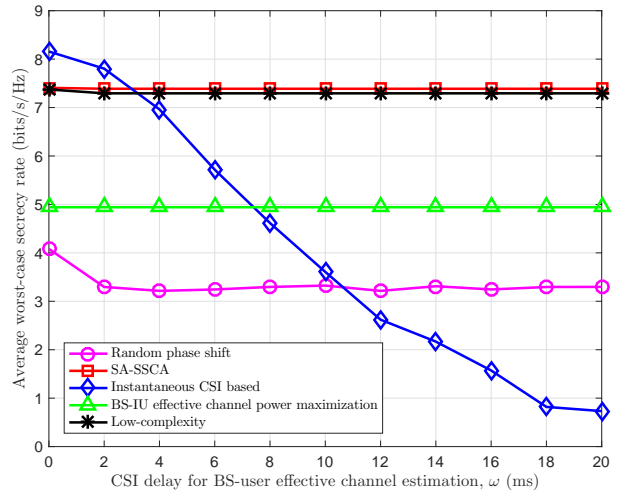}}
	\caption{Average worst-case secrecy rate versus CSI delay for BS-user effective channel estimation, $\omega$.}\label{fig5}
\end{figure}

{\color{black}
\subsection{Impact of Imperfect Statistical CSI}
In Fig. \ref{fig_final}, we study the impact of imperfect statistical CSI on the worst-case secrecy rate. For simplicity, we assume that the channel samples generated are not perfect and the channel errors follow complex Gaussian distributions. For example, the BS-IU channel samples $\hat{\mathbf{h}}_{1,j}$ are generated according to $\hat{\mathbf{h}}_{1,j} = \mathbf{h}_{1,j} + \Delta \mathbf{h}_{1,j}$, where $\Delta \mathbf{h}_{1,j} \in \mathcal{CN}(\mathbf{0}, \delta_{1,j}^2 \mathbf{I})$, $\delta_{1,j} = \sqrt{b_{1,j} L^{\textrm{BI}}}$, and $b_{1,j}$ is a parameter controlling the error variance. The other channel samples are generated similarly,  and we assume that the control parameters of all channel errors are the same (denoted by $b$). It is observed that the performance of the proposed SA-SSCA and low-complexity algorithms does not change much in the presence of statistical CSI errors, and the average worst-case secrecy rate decreases only marginally  (less than 1$\%$) when $b$ is higher than $-10$ dB. Therefore, the proposed algorithms are robust against statistical CSI errors having Gaussian distributions. 
}

{\color{black}
	\subsection{Impact of Time-Variant Statistical CSI} \label{section_time_variant_S_CSI}
	Finally, we show the effect of rapidly-fluctuation channel statistics on the performance of the proposed SA-SSCA algorithm, which is depicted in Fig. \ref{fig_final_S_CSI}. In this simulation, we assume that the channel statistics mainly change owing to the user mobility and the IU is moving from $(6\;\textrm{m}, 200\;\textrm{m}, 0)$ to $((6+ x_{mo})\;\textrm{m}, 200\;\textrm{m}, 0)$ along the $x$ axis. When the IU moves, the LoS components will change, since the azimuth and elevation angles between the IU and BS/RIS dynamically fluctuate. Note that in practice the time-variant channel statistics might be more complex, but the model considered is still useful since it captures the key feature that the channel statistics gradually decorrelate with time. For the proposed SA-SSCA algorithm, we consider two cases. In the first case the long-term phase shifts are optimized using outdated channel statistics obtained when the IU is located at $(6\;\textrm{m}, 200\;\textrm{m}, 0)$, while in the second case updated channel statistics at $((6+ x_{mo})\;\textrm{m}, 200\;\textrm{m}, 0)$ are available. As seen from Fig. \ref{fig_final_S_CSI}, the performance of all the algorithms considered degrades upon increasing $x_{mo}$, which is expected since the distances between the IU and BS/RIS increase with $x_{mo}$. Furthermore, one can observe that using outdated channel statistics only leads to minor performance degradations (about 0.1 bits/s/Hz when $x_{mo} \leq 170\;\textrm{m}$). This means that the proposed SA-SSCA algorithm is quite robust against time-variant statistical CSI variations under the model considered.
}

\begin{figure}[t]
	\centering
	\setlength{\abovecaptionskip}{-0cm}
	\setlength{\belowcaptionskip}{-0cm}
	\scalebox{0.48}{\includegraphics{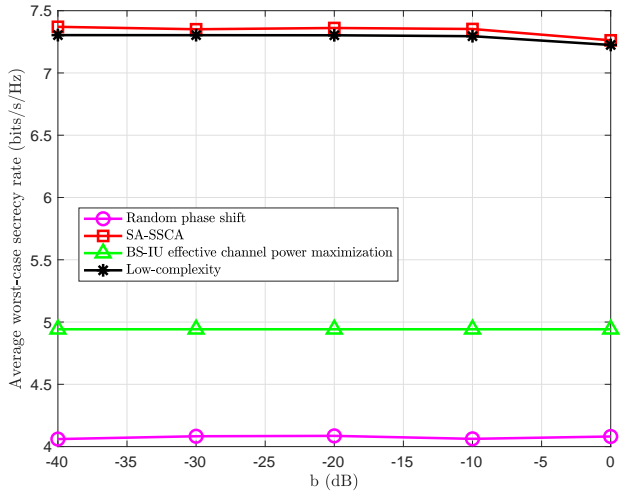}}
	\caption{Average worst-case secrecy rate versus $b$.}\label{fig_final}
\end{figure}

\begin{figure}[t]
	\centering
	\setlength{\abovecaptionskip}{-0cm}
	\setlength{\belowcaptionskip}{-0cm}
	\scalebox{0.48}{\includegraphics{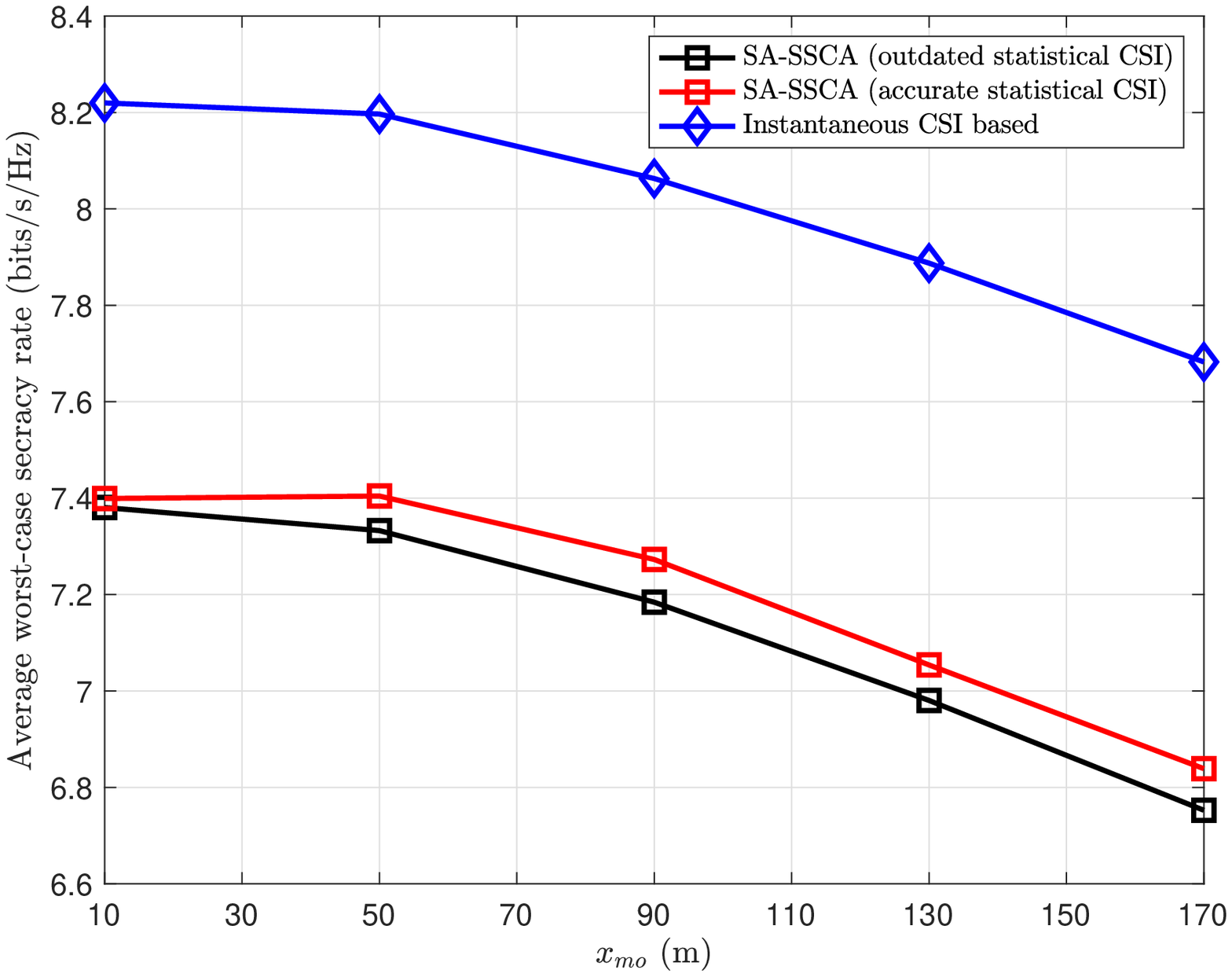}}
	\caption{Average worst-case secrecy rate versus $x_{mo}$.}\label{fig_final_S_CSI}
\end{figure}

\section{Conclusions} \label{sec_conclusions}
In this paper, we investigated a new two-timescale-based beamforming optimization problem for secrecy rate maximization in a RIS-aided SWIPT system, where the EUs are also potential eavesdroppers. Specifically, the short-term active beamforming vectors at the BS (including information and energy beams) and the long-term passive reflecting phase shifts at the RIS are jointly optimized to maximize the average worst-case secrecy rate subject to both the power constraint at the BS and energy harvesting constraints at the EUs. We proposed two algorithms (namely the SA-SSCA algorithm and the low-complexity algorithm) to achieve a balance between the secrecy performance and computational complexity. In particular, the proposed SA-SSCA algorithm was shown to achieve better performance (with guaranteed convergence) as compared with other benchmark schemes, while the low-complexity algorithm is able to provide competitive performance with reduced complexity. Simulation results validated the effectiveness of RIS for guaranteeing security as well as enhancing SWIPT performance. Besides, it was found that the two-timescale transmission scheme is feasible for secure SWIPT systems, while the required channel estimation overhead is significantly reduced.

\begin{appendices}
	
\section{Proof of Theorem \ref{trans_theorem1}} \label{proof_theorem1}
Since the introduced auxiliary variables $z$, $u$ and $v$ only appear in the objective function of problem \eqref{trans_prob2}, the optimal solution of problem \eqref{trans_prob2} must satisfy the following first-order optimality conditions:
\begin{equation}\label{zuv_solution}
\begin{aligned}
z^{\star}&=\frac{1}{e(\mathbf{w}^{\star}, \mathbf{P}^{\star},u^{\star})},\;
v^{\star}=\frac{1}{\tilde{e}(\mathbf{w}^{\star},\mathbf{P}^{\star})},\\
u^{\star} & = \frac{\tilde{\mathbf{h}}^H\mathbf{w}^{\star}}{|\tilde{\mathbf{h}}^H\mathbf{w}^{\star}|^2+\sigma^2+\|\tilde{\mathbf{h}}^H\mathbf{P}^{\star}\|^2}.
\end{aligned}
\end{equation}
Then, by substituting \eqref{zuv_solution} back into the objective function of problem \eqref{trans_prob2}, we can readily obtain $\mathcal{P}_S(\bm{\theta}, \mathcal{H})$. This thus completes the proof.

\section{Proof of Theorem \ref{theorem_convergence}} \label{proof_of_convergence}
First, for notational simplicity, we define the following functions:
\begin{equation}
	\begin{aligned}
&	\bar{E}(\mathbf{w},\mathbf{P}, z, u, v, y_m)   \triangleq  ze(\mathbf{w},\mathbf{P},u)-\log(z)\\
	&\quad \quad \quad \quad \quad +\frac{1}{p}v\left(\sum_{m \in \mathcal{M}}(1+y_m)^{p}\right)-\frac{1}{p}\log(v),\\
& \tilde{g}_m(y_m, \mathbf{w},\mathbf{P}; \mathbf{P}^{f})  \\
& \quad \triangleq \frac{|\tilde{\mathbf{g}}_m^H\mathbf{w}|^2}{\!-\! \|\tilde{\mathbf{g}}_{m}^H\mathbf{P}^{f}\|^2 \!+\! 2\Re\{\tilde{\mathbf{g}}_{m}^H\mathbf{P}^{f} \mathbf{P}^H\tilde{\mathbf{g}}_{m} \}\!+\! \sigma_{E,m}^2} \!-\! y_m,\\
& g_m(y_m, \mathbf{w}, \mathbf{P})  \triangleq  \textrm{SINR}_m^{\textrm{E}}  (\mathbf{w}, \mathbf{P}, \bm{\theta}) - y_m,\\
&	\tilde{\varepsilon}_m(\mathbf{w}, \mathbf{P};\mathbf{w}^{f}, \mathbf{P}^{f})  \triangleq   \epsilon_m+|\tilde{\mathbf{g}}_{m}^H\mathbf{w}^{f}|^2+\|\tilde{\mathbf{g}}_{m}^H\mathbf{P}^{f}\|^2 \\
	& \quad \quad \quad \quad \quad \quad - 2\Re\{\tilde{\mathbf{g}}_{m}^H\mathbf{w}^{f}\mathbf{w}^H\tilde{\mathbf{g}}_{m} +\tilde{\mathbf{g}}_{m}^H\mathbf{P}^{f}\mathbf{P}^H\tilde{\mathbf{g}}_{m} \},\\
&	\varepsilon_m(\mathbf{w}, \mathbf{P})  \triangleq  \epsilon_m - |\tilde{\mathbf{g}}_{m}^H\mathbf{w}|^2-\|\tilde{\mathbf{g}}_{m}^H\mathbf{P}\|^2,\\
&	P_{c}(\mathbf{w}, \mathbf{P})  \triangleq  \|\mathbf{w}\|^2+\|\mathbf{P}\|^2,
	\end{aligned}
\end{equation}
where $f(x;y)$ means that $x$ is a variable of the function $f(\cdot;\cdot)$ and $y$ is a given parameter.\footnote{In the following, we may drop the arguments in the function $f(x;y)$, and simply use $f$ in the sequel of this paper when there is no ambiguity.}
Then, we prove that every iterate $\{ \mathbf{w}^{i}, \mathbf{P}^{i},z^{i}, u^{i}, v^{i}\}$ generated by Algorithm \ref{inner_algorithm} is feasible to problem \eqref{trans_prob3}. Suppose that $\{ \mathbf{w}^{i}, \mathbf{P}^{i},z^{i}, u^{i}, v^{i}\}$ is feasible to problem \eqref{trans_prob3}, then according to the inequalities in \eqref{ineq_snr} and \eqref{ineq_eh}, we have $\varepsilon_m(\mathbf{w}^{i+1}, \mathbf{P}^{i+1})\leq\tilde{\varepsilon}_m(\mathbf{w}^{i+1}, \mathbf{P}^{i+1};\mathbf{w}^{i}, \mathbf{P}^{i})\leq 0$ and $g_m(y_m^{i+1}, \mathbf{w}^{i+1}, \mathbf{P}^{i+1})\leq \tilde{g}_m(y_m^{i+1}, \mathbf{w}^{i+1},\mathbf{P}^{i+1}; \mathbf{P}^{i})\leq 0$, which implies that $\{ \mathbf{w}^{i+1}, \mathbf{P}^{i+1},z^{i+1}, u^{i+1},$ $ v^{i+1}\}$ is feasible to problem \eqref{trans_prob3}. Therefore, if Algorithm \ref{inner_algorithm} is initialized with a feasible solution, the subsequent solutions generated by Algorithm \ref{inner_algorithm} are all feasible to problem \eqref{trans_prob3}.

Next, we show that the sequence of the objective values generated by Algorithm \ref{inner_algorithm} is non-increasing over the iterations. Specifically, it is readily seen that the global optimal solutions of problems \eqref{step1_problem}, \eqref{step2_problem} and \eqref{step3_problem} can be obtained, i.e., each step of Algorithm \ref{inner_algorithm} minimizes $\bar{E}$ with the other variables being fixed. Thus, we have $\bar{E}(\mathbf{w}^{i},\mathbf{P}^{i}, z^{i}, u^{i}, v^{i}, y_m^{i})\geq \bar{E}(\mathbf{w}^{i},\mathbf{P}^{i}, z^{i+1}, u^{i}, v^{i+1}, y_m^{i})\geq \bar{E}(\mathbf{w}^{i},\mathbf{P}^{i}, z^{i+1}, u^{i+1}, v^{i+1}, y_m^{i})\geq \bar{E}(\mathbf{w}^{i+1},\mathbf{P}^{i+1}, z^{i+1} $ $, u^{i+1}, v^{i+1}, y_m^{i+1})$. Consequently, we can prove the monotonic convergence of the sequence $\{\bar{E}(\mathbf{w}^{i},\mathbf{P}^{i}, z^{i}$, $u^{i}, v^{i}, y_m^{i})\}$.

Then, we are ready to prove that any limit point $\{\mathbf{w}^{\star},\mathbf{P}^{\star}, z^{\star}, u^{\star}, v^{\star}, y_m^{\star}\}$ of the iterates generated by Algorithm \ref{inner_algorithm} is a stationary solution of problem \eqref{trans_prob3}. To proceed, we first give the following lemma.
\begin{lemma}\label{lemma 2}
Let $f(\mathbf{x})= g(y(\mathbf{x}),\mathbf{x})$, we have that if $\tilde{y}(\mathbf{x};\mathbf{x}^f)|_{\mathbf{x}=\mathbf{x}^f}=y(\mathbf{x}^f)$ and $\frac{\partial \tilde{y}(\mathbf{x};\mathbf{x}^f)}{\partial \mathbf{x}}|_{\mathbf{x}=\mathbf{x}^f}=\frac{d y(\mathbf{x})}{d\mathbf{x}}|_{\mathbf{x}=\mathbf{x}^f}$, then $\frac{d f(\mathbf{x})}{d\mathbf{x}}|_{\mathbf{x}=\mathbf{x}^f}=\frac{d g(\tilde{y}(\mathbf{x};\mathbf{x}^f),\mathbf{x})}{d \mathbf{x}}|_{\mathbf{x}=\mathbf{x}^f}$.
\end{lemma}
\begin{proof}
Based on the chain rule, we obtain
\begin{equation}
\begin{aligned}
	 & \frac{d g(\tilde{y}(\mathbf{x}; \mathbf{x}^f),\mathbf{x})}{d \mathbf{x}}\bigg|_{\mathbf{x}=\mathbf{x}^f}\\
	=&\frac{\partial g(\tilde{y},\mathbf{x})}{\partial \tilde{y}}\frac{\partial \tilde{y}(\mathbf{x};\mathbf{x}^f)}{\partial\mathbf{x}}\bigg|_{\mathbf{x}=\mathbf{x}^f,\tilde{y}=\tilde{y}(\mathbf{x}^f;\mathbf{x}^f)}\\
	& +\frac{\partial g(\tilde{y},\mathbf{x})}{\partial\mathbf{x}}\bigg|_{\mathbf{x}=\mathbf{x}^f,\tilde{y}=\tilde{y}(\mathbf{x}^f;\mathbf{x}^f)}\\
	=&\frac{\partial g(y,\mathbf{x})}{\partial y}\frac{\partial y(\mathbf{x})}{\partial\mathbf{x}}\bigg|_{\mathbf{x}=\mathbf{x}^f,y=y(\mathbf{x}^f)} \!+\! \frac{\partial g(y,\mathbf{x})}{\partial\mathbf{x}}\bigg|_{\mathbf{x}=\mathbf{x}^f,y=y(\mathbf{x}^f)}\\
	=&\frac{d g(y(\mathbf{x}),\mathbf{x})}{d \mathbf{x}}\bigg|_{\mathbf{x}=\mathbf{x}^f} = \frac{d f(\mathbf{x})}{d\mathbf{x}}\bigg|_{\mathbf{x}=\mathbf{x}^f},
\end{aligned}
\end{equation}
where the second equality holds because $\tilde{y}(\mathbf{x};\mathbf{x}_1)|_{\mathbf{x}=\mathbf{x}_1}=y(\mathbf{x}_1)$ and $\frac{\partial \tilde{y}(\mathbf{x};\mathbf{x}_1)}{\partial \mathbf{x}}|_{\mathbf{x}=\mathbf{x}_1}=\frac{d y(\mathbf{x})}{d\mathbf{x}}|_{\mathbf{x}=\mathbf{x}_1}$. This thus completes the proof.
\end{proof}
\noindent
Based on Lemma \ref{lemma 2}, it follows that $\frac{\partial g_m}{\partial \mathbf{w}}\big|_{\mathbf{P}=\mathbf{P}^{f}}=\frac{\partial \tilde{g}_m}{\partial \mathbf{w}}\big|_{\mathbf{P}=\mathbf{P}^{f}}$ holds, which can be simply obtained by regarding $-\|\tilde{\mathbf{g}}_{m}^H\mathbf{P}^{f}\|^2+ 2\Re\{\tilde{\mathbf{g}}_{m}^H\mathbf{P}^{f}\mathbf{P}^H\tilde{\mathbf{g}}_{m} \}$ and $\|\tilde{\mathbf{g}}_{m}^H\mathbf{P}\|^2$ as $\tilde{y}$ and $y$ in Lemma \ref{lemma 2}, respectively.
Similarly, we can obtain $\frac{\partial\varepsilon_m}{\partial \mathbf{w}}\big|_{\mathbf{w}=\mathbf{w}^{f}, \mathbf{P} =\mathbf{P}^{f}} \!=\! \frac{\partial\tilde{\varepsilon}_m}{\partial \mathbf{w}}\big|_{\mathbf{w}=\mathbf{w}^{f}, \mathbf{P}=\mathbf{P}^{f}}$, $\frac{\partial\varepsilon_m}{\partial \mathbf{P}}\big|_{\mathbf{w}=\mathbf{w}^{f}, \mathbf{P}=\mathbf{P}^{f}} \!=\! \frac{\partial\tilde{\varepsilon}_m}{\partial \mathbf{P}}\big|_{\mathbf{w}=\mathbf{w}^{f}, \mathbf{P}=\mathbf{P}^{f}}$, $\frac{\partial g_m}{\partial \mathbf{P}}\big|_{\mathbf{P}=\mathbf{P}^{f}} \!=\! \frac{\partial \tilde{g}_m}{\partial \mathbf{P}}\big|_{\mathbf{P}=\mathbf{P}^{f}}$ and $\frac{\partial g_m}{\partial y_m}\big|_{\mathbf{P}=\mathbf{P}^{f}}=\frac{\partial \tilde{g}_m}{\partial y_m}\big|_{\mathbf{P}=\mathbf{P}^{f}}$.
Furthermore, let $\mathbf{S}\triangleq(\mathbf{w},\mathbf{P}, z, u, v, y_m; \mathbf{w}^{f},\mathbf{P}^{f})$ denote the composition of the optimization variables and parameters $\mathbf{w}^{f}$ and $\mathbf{P}^{f}$, and let $\mathbf{S}^\star \triangleq(\mathbf{w}^\star,\mathbf{P}^\star, z^\star, u^\star, v^\star, y_m^\star; \mathbf{w}^{\star},\mathbf{P}^{\star})$ denote the limit point of the sequence $\mathbf{S}^i$ generated by Algorithm \ref{inner_algorithm}. Then, based on abovementioned facts about the gradients and by checking the first-order optimality conditions of problems \eqref{step1_problem}-\eqref{step3_problem}, we have the following equations:
\begin{equation}\label{first_order_start}
	\frac{\partial\bar{E}}{\partial z}\bigg|_{\mathbf{S}=\mathbf{S}^{\star}}=0,\;\frac{\partial \bar{E}}{\partial v}\bigg|_{\mathbf{S}=\mathbf{S}^{\star}}=0,\;
	\frac{\partial\bar{E}}{\partial u}\bigg|_{\mathbf{S}=\mathbf{S}^{\star}}=0,
\end{equation}
\begin{equation}
	\begin{aligned}
	\sum_{m\in \mathcal{M}} & \bigg(\frac{\partial \bar{E}}{\partial y_m} \bigg|_{\mathbf{S}=\mathbf{S}^{\star}} + \lambda_m\frac{\partial \tilde{g}_m}{\partial y_m}\bigg|_{\mathbf{S}=\mathbf{S}^{\star}}\bigg)\\
	&  =\sum_{m\in \mathcal{M}}\bigg(\frac{\partial \bar{E}}{\partial y_m}\bigg|_{\mathbf{S}=\mathbf{S}^{\star}} + \lambda_m\frac{\partial g_m}{\partial y_m}\bigg|_{\mathbf{S}=\mathbf{S}^{\star}}\bigg)=0,
	\end{aligned}
\end{equation}
\begin{equation}
	\begin{aligned}
	\sum_{m\in \mathcal{M}}& \bigg(\frac{\partial \bar{E}}{\partial \mathbf{w}}\bigg|_{\mathbf{S}=\mathbf{S}^{\star}} + \lambda_m\frac{\partial \tilde{g}_m}{\partial \mathbf{w}}\bigg|_{\mathbf{S} =\mathbf{S}^{\star}}\bigg)\\
	& =\sum_{m\in \mathcal{M}}\bigg(\frac{\partial \bar{E}}{\partial \mathbf{w}}\bigg|_{\mathbf{S}=\mathbf{S}^{\star}} + \lambda_m\frac{\partial g_m}{\partial \mathbf{w}}\bigg|_{\mathbf{S}=\mathbf{S}^{\star}}\bigg)=\mathbf{0},
	\end{aligned}
\end{equation}
\begin{equation}
	\begin{aligned}
	\sum_{m \in \mathcal{M}}& \bigg(\frac{\partial \bar{E}}{\partial \mathbf{P}}\bigg|_{\mathbf{S}=\mathbf{S}^{\star}} + \lambda_m\frac{\partial \tilde{g}_m}{\partial \mathbf{P}}\bigg|_{\mathbf{S}=\mathbf{S}^{\star}}\bigg)\\
	& =\sum_{m \in \mathcal{M}}\bigg(\frac{\partial \bar{E}}{\partial \mathbf{P}}\bigg|_{\mathbf{S}=\mathbf{S}^{\star}} + \lambda_m\frac{\partial g_m}{\partial \mathbf{P}}\bigg|_{\mathbf{S}=\mathbf{S}^{\star}}\bigg)=\mathbf{0},
	\end{aligned}
\end{equation}
\begin{equation}
	\begin{aligned}
	\sum_{m\in \mathcal{M}}& \bigg(\frac{\partial \bar{E}}{\partial \mathbf{w}}\bigg|_{\mathbf{S}=\mathbf{S}^{\star}} + \gamma_m\frac{\partial \tilde{\varepsilon}_m}{\partial \mathbf{w}}\bigg|_{\mathbf{S}=\mathbf{S}^{\star}}\bigg)\\
	& =\sum_{m \in \mathcal{M}}\bigg(\frac{\partial \bar{E}}{\partial \mathbf{w}}\bigg|_{\mathbf{S}=\mathbf{S}^{\star}} + \gamma_m\frac{\partial \varepsilon_m}{\partial \mathbf{w}}\bigg|_{\mathbf{S}=\mathbf{S}^{\star}}\bigg)=\mathbf{0},
	\end{aligned}
\end{equation}
\begin{equation}
	\begin{aligned}
	\sum_{m\in \mathcal{M}}& \bigg(\frac{\partial \bar{E}}{\partial \mathbf{P}}\bigg|_{\mathbf{S}=\mathbf{S}^{\star}} + \gamma_m\frac{\partial \tilde{\varepsilon}_m}{\partial \mathbf{P}}\bigg|_{\mathbf{S}=\mathbf{S}^{\star}}\bigg)\\
	& =\sum_{m\in \mathcal{M}}\bigg(\frac{\partial \bar{E}}{\partial \mathbf{P}}\bigg|_{\mathbf{S}=\mathbf{S}^{\star}} + \gamma_m\frac{\partial \varepsilon_m}{\partial \mathbf{P}}\bigg|_{\mathbf{S}=\mathbf{S}^{\star}}\bigg)=\mathbf{0},
	\end{aligned}
\end{equation}
\begin{equation}\label{first_order_end}
	\begin{aligned}
	& \frac{\partial \bar{E}}{\partial \mathbf{w}}\bigg|_{\mathbf{S}=\mathbf{S}^{\star}} + \delta\frac{\partial P_{c}}{\partial \mathbf{w}}\bigg|_{\mathbf{S}=\mathbf{S}^{\star}}=\mathbf{0},\\
	& \frac{\partial \bar{E}}{\partial \mathbf{P}}\bigg|_{\mathbf{S}=\mathbf{S}^{\star}} + \delta\frac{\partial P_{c}}{\partial \mathbf{P}}\bigg|_{\mathbf{S}=\mathbf{S}^{\star}}=\mathbf{0},
	\end{aligned}
\end{equation}
where $\{\lambda_m\}$, $\{\gamma_m\}$ and $\delta$ are the Lagrange multipliers associated with the constraints \eqref{SINR_constraint1}, \eqref{EH_constraint1} and \eqref{power_constraint}, respectively, which satisfy
\begin{equation}\label{dual_feas}
	\lambda_m\geq 0,\; \gamma_m\geq 0,\;\forall m \in \mathcal{M},\; \delta\geq 0.
\end{equation}
Moreover, based on the complementary slackness condition of problems \eqref{step1_problem}-\eqref{step3_problem}, i.e., $\lambda_m $ $\tilde{g}_m(y_m^{\star},\mathbf{w}^{\star},\mathbf{P}^{\star}; \mathbf{P}^{\star})=0$, $
\gamma_m\tilde{\varepsilon}_m(\mathbf{w}^{\star},\mathbf{P}^{\star}; \mathbf{w}^{\star},\mathbf{P}^{\star})=0, \forall m \in \mathcal{M}$, $ \delta P_{c}=0$, and the fact $\tilde{g}_m(y_m^{\star},$ $\mathbf{w}^{\star},\mathbf{P}^{\star}; \mathbf{P}^{\star})=g_m(y_m^{\star},\mathbf{w}^{\star},\mathbf{P}^{\star})$, $\tilde{\varepsilon}_m(\mathbf{w}^{\star},\mathbf{P}^{\star}; \mathbf{w}^{\star},\mathbf{P}^{\star})=\varepsilon_m(\mathbf{w}^{\star},\mathbf{P}^{\star}), \forall m \in \mathcal{M}$, we have
\begin{equation}\label{slack_condition}
	\begin{aligned}
		&\lambda_m\tilde{g}_m(y_m^{\star},\mathbf{w}^{\star},\mathbf{P}^{\star}; \mathbf{P}^{\star})=\lambda_mg_m(y_m^{\star},\mathbf{w}^{\star},\mathbf{P}^{\star})=0,\\
		&\gamma_m\tilde{\varepsilon}_m(\mathbf{w}^{\star},\mathbf{P}^{\star}; \mathbf{w}^{\star},\mathbf{P}^{\star})=\gamma_m\varepsilon_m(\mathbf{w}^{\star},\mathbf{P}^{\star})=0, \\
		& \forall m \in \mathcal{M}, \; \delta P_{c}=0.
	\end{aligned}
\end{equation}

To summarize, we can see that \eqref{first_order_start}-\eqref{first_order_end} are the first-order optimality conditions of problem \eqref{trans_prob3}, while \eqref{dual_feas} and \eqref{slack_condition} are respectively the dual feasibility and  complementary slackness conditions of problem \eqref{trans_prob3}. Together with the primal feasibility proved before, we can conclude that the limit point $\mathbf{S}^\star$ is a stationary solution of problem \eqref{trans_prob3}. Finally, based on Theorem 3 in \cite{Shi2011WMMSE}, we can infer that any limit point generated by Algorithm \ref{inner_algorithm} is also a stationary solution of $\mathcal{P}_S(\bm{\theta}, \mathcal{H})$, which thus completes the proof.

\section{Derivation of the partial derivatives}\label{derivative}
Since $\bm{\theta} = \angle \bm{\phi}$, we can express the gradient of $\bar{S}$ w.r.t. $\boldsymbol{\theta}$, i.e., $\frac{\partial \bar{S}}{\partial \bm{\theta}}$, as a function of $\frac{\partial \bar{S}}{\partial \bm{\phi}}$, which is shown as follows:
\begin{equation}\label{gradient_theta}
	\frac{\partial \bar{S}}{\partial \bm{\theta}}=\frac{\partial \bar{S}}{\partial \boldsymbol{\phi}}\circ \jmath \boldsymbol{\phi}-\frac{\partial \bar{S}}{\partial \boldsymbol{\phi}^*}\circ \jmath \boldsymbol{\phi}^*.
\end{equation}
Then, we can easily obtain $\frac{\partial \bar{S}}{\partial \bm{\phi}}$ as in \eqref{gradiant_phi} (shown at the top of this page)
\begin{figure*}[ht]
	\begin{equation}\label{gradiant_phi}
		\begin{aligned}
			\frac{\partial \bar{S}}{\partial \boldsymbol{\phi}} =& \frac{1}{1+\textrm{SINR}^{\textrm{I}}(\mathbf{w}, \mathbf{P}, \bm{\theta})} \left(\frac{\big((\|\tilde{\mathbf{h}}^H\mathbf{P}\|^2+\sigma_{I}^2)\tilde{\mathbf{h}}^H\mathbf{ww}^H-|\tilde{\mathbf{h}}^H\mathbf{w}|^2\tilde{\mathbf{h}}^H\mathbf{PP}^H\big)\mathbf{F}_1\text{diag}(\mathbf{h}_2)}{(\|\tilde{\mathbf{h}}^H\mathbf{P}\|^2+\sigma_{I}^2)^2}\right)^T\\
			& -\frac{1}{\sum_{m\in \mathcal{M}}(1+\textrm{SINR}_m^{\textrm{E}}(\mathbf{w}, \mathbf{P}, \bm{\theta}))^p}\sum_{m \in \mathcal{M}}\bigg((1+\textrm{SINR}_m^{\textrm{E}}(\mathbf{w}, \mathbf{P}, \bm{\theta}))^{p-1}\\
			& \times \frac{\big( (\|\tilde{\mathbf{g}}_m^H\mathbf{P}\|^2+\sigma_{E,m}^2)\tilde{\mathbf{g}}_m^H\mathbf{ww}^H-|\tilde{\mathbf{g}}_m^H\mathbf{w}|^2\tilde{\mathbf{g}}_m^H\mathbf{PP}^H \big)\mathbf{F}_1\text{diag}(\mathbf{g}_{2,m})}{(\|\tilde{\mathbf{g}}_m^H\mathbf{P}\|^2+\sigma_{E,m}^2)^2}\bigg)^T,
		\end{aligned}
	\end{equation}
	\hrulefill
\end{figure*}
and similarly, we can obtain $\frac{\partial \bar{S}}{\partial \bm{\phi}^*}$. By substituting $\frac{\partial \bar{S}}{\partial \bm{\phi}}$ and $\frac{\partial \bar{S}}{\partial \bm{\phi}^*}$ into \eqref{gradient_theta}, we can finally have the expression of $\frac{\partial \bar{S}}{\partial \bm{\theta}}$ and the value of $\frac{\partial \bar{S}}{\partial \boldsymbol{\theta}}\big|_{(\mathbf{w}^J(\bm{\theta}^t,\mathcal{H}_j^t), \mathbf{P}^J(\bm{\theta}^t,\mathcal{H}_j^t),\boldsymbol{\theta}^t, \mathcal{H}^t_j)}$ can be obtained accordingly.
\end{appendices}

\bibliographystyle{IEEEtran}
\bibliography{bib}

\end{document}